\documentclass[journal,twoside,web]{ieeecolor}
\usepackage{generic}
\usepackage{amsmath,amssymb,amsfonts}
\usepackage{mathtools}
\usepackage{algorithm}
\usepackage{algpseudocode}
\usepackage{graphicx}
\usepackage{textcomp}
\usepackage{caption}
\usepackage{subcaption}
\usepackage[normalem]{ulem}
\usepackage{balance}
\usepackage[
backend=biber,
style=ieee,
citestyle=ieee,
sorting=none
]{biblatex}

\newtheorem{theorem}{Theorem}
\newtheorem{lemma}{Lemma} 
\newtheorem{assumption}{Assumption}
\newtheorem{definition}{Definition} 
\newtheorem{proposition}{Proposition} 
\newtheorem{remark}{Remark}

\newcommand{\R}{\mathbb{R}} 
\newcommand{\node}{\mathcal{V}}     
\newcommand{\edge}{\mathcal{E}}     
\newcommand{\gd}{\mathcal{G}_t}     
\newcommand{\bd}{\mathcal{B}_t}     
\newcommand{\nei}{\mathcal{N}}      
\newcommand{\kg}{K_\mathcal{G}}
\newcommand{\kb}{K_\mathcal{B}} 
\newcommand{\linf}{L_\infty}
\newcommand{\winf}{W_\infty}
\newcommand{\dout}{D^\text{out}} 
\newcommand{\degInMax}{d^{\text{in}}_{\text{max}}}
\newcommand{\degOutMax}{d^{\text{out}}_{\text{max}}}

\newcommand{\ts}[1]{\theta^*(#1)}
\newcommand{\tti}{t \rightarrow \infty}
\newcommand{\gmo}[1]{ \gamma_1 (#1) }
\newcommand{\gmt}[1]{ \gamma_2 (#1) }

\newcommand{\nrm}[1]{\Vert #1 \Vert} 
\newcommand{\nrmsq}[1]{\Vert #1 \Vert^2}
\newcommand{\bbc}[1]{\big( #1 \big)} 
\newcommand{\Bbc}[1]{\Big( #1 \Big)} 
\newcommand{\one}{\mathbf{1}} 
\newcommand{\zro}{\mathbf{0}}
\newcommand{\J}{\frac{1}{N}\mathbf{1}\mathbf{1}^T}

\def\BibTeX{{\rm B\kern-.05em{\sc i\kern-.025em b}\kern-.08em
    T\kern-.1667em\lower.7ex\hbox{E}\kern-.125emX}}
\begin{document}
\title{Distributed Estimation over Directed Graphs Resilient to Sensor Spoofing}
\author{Shamik Bhattacharyya, Kiran Rokade, and Rachel Kalpana Kalaimani
\thanks{This work has been partially supported by DST-INSPIRE Faculty Grant, Department of Science and Technology (DST), Govt. of India (ELE/16-17/333/DSTX/RACH)}
\thanks{Shamik Bhattacharyya is with the Electrical Engineering Department, Indian Institute of Technology Madras, TN 600036, India (e-mail: ee18d005@smail.iitm.ac.in). }
\thanks{Kiran Rokade is with the Electrical and Computer Engineering Department, Cornell University, Ithaca, NY 14850, USA (e-mail: kvr36@cornell.edu).}
\thanks{Rachel K. Kalaimani is with 
the Electrical Engineering Department, Indian Institute of Technology Madras, TN 600036, India (e-mail: rachel@ee.iitm.ac.in).}}

\maketitle

\begin{abstract}
This paper addresses the problem of distributed estimation of an unknown dynamic parameter by a multi-agent system over a directed communication network in the presence of an adversarial attack on the agents' sensors. The mode of attack of the adversaries is to corrupt the sensor measurements of some of the agents, while the communication and information processing capabilities of those agents remain unaffected. To ensure that all the agents, both normal as well as those under attack, are able to correctly estimate the parameter value, the Resilient Estimation through Weight Balancing (REWB) algorithm is introduced. The only condition required for the REWB algorithm to guarantee resilient estimation is that at any given point in time, less than half of the total number of agents are under attack. The paper discusses the development of the REWB algorithm using the concepts of weight balancing of directed graphs, and the consensus+innovations approach for linear estimation. Numerical simulations are presented to illustrate the performance of our algorithm over directed graphs under different conditions of adversarial attacks.
\end{abstract}


\section{Introduction}
\label{sec:introduction}
\IEEEPARstart{T}{he} advancement in wireless sensor networks (WSNs) has diversified their areas of  application to agriculture \cite{AppAgri}, healthcare \cite{AppHlth}, and renewable energy \cite{AppREnrg} to name a few. As a result the scale and complexity of the networks is also on the rise \cite{SenSrvey}, \cite{WSNSrvey}. This necessitates the use of more distributed approaches to signal processing over WSNs, and distributed estimation is a key aspect of it. Distributed estimation is about determining a parameter of interest locally at each sensor node with cooperation between neighboring nodes \cite{DEstiSrvey}. The increase in areas of application of WSNs has in turn made them more vulnerable to adversarial attacks \cite{Security}. A major mode of such attacks are aimed to manipulate the normal functioning of the sensor nodes and thus disrupt the overall signal processing capability of the WSNs. Some commonly used threat models are Byzantine \cite{Thrt_Byz}, malicious \cite{Thrt_Mal}, sensor spoofing \cite{SenSpoof_TAC16}, etc. Hence the distributed estimation algorithms need to be resilient to adversarial attacks in order to be more effective.

Different consensus algorithms resilient to adversarial attacks appear in the literature such as Mean-Subsequence-Reduced algorithm \cite{Dbji_Atc17} and Median Consensus Algorithm \cite{MCA}. These algorithms ensure consensus only for the normal agents, while the agents under attack may have arbitrary values. We are interested in a resilient distributed estimation algorithm, that will ensure that both the normal agents and the agents under attack can reach consensus over the true value of the parameter to be estimated. 
The \emph{consensus+innovations} approach illustrated in \cite{KarSPMag13} uses the consensus framework to design resilient algorithms for linear estimation. The Constant weight Saturated Innovation Update (CSIU) algorithm \cite{KarCDC18} is one such resilient estimation algorithm which ensures that all the agents are able to estimate the parameter of interest, provided less than three-tenth of the total agents are under attack. This was further improved in \cite{KarTAC19}, where the Saturated Innovation Update (SIU) algorithm ensures all the agents' estimate converge to the desired parameter value provided the adversaries attack less than half of the total agents. Also in \cite{KarCDC18}, a new term \emph{resilience index} was used to provide a bound for the fraction of sensor nodes under attack. 

Both the CSIU and SIU algorithms are designed on undirected graphs representing bidirectional communication links between the agents. In many practical scenarios, the power levels at which sensor nodes broadcast information or, their interference and noise patterns, differ from node to node \cite{DiGrph}, \cite{BalWts}. The communication between nodes in such cases is unidirectional which is aptly represented by a directed graph. Here we consider a time-invariant network topology with unidirectional communication links between agents. To the best of our knowledge, an extension of the consensus+innovations approach to directed graphs is non-existent in the literature, except for the recent work \cite{Meng}. However, we observed that for the algorithm presented in \cite{Meng}, choosing appropriate parameters is not an easy task. In contrast, we propose an algorithm which guarantees convergence over a given range of parameter values. Also, unlike \cite{Meng}, where  the set of adversarial agents and the unknown parameter are fixed, our proposed algorithm works even when the set of agents under attack and the unknown parameter changes with time.

 The model of attack by the adversaries is designed on the idea of sensor spoofing \cite{SenSpoof_TAC16} where the sensor readings of the agents under attack are corrupted through data falsification or false data injection. Note that such an attack on the agents is restricted to their sensors. In particular, the agents under attack can perform computations and communicate with their neighbours. Also the agents under attack by the adversaries are not known a-priori by the normal agents. Moreover we allow for a more general scenario where the adversaries may attack different agents over time. We present an algorithm, Resilient Estimation through Weight Balancing (REWB), which ensures that all agents asymptotically converge to the value to be estimated provided less than half of the total number of agents are affected by adversaries. The agents operate in a distributed manner using only the local information available to them. The main contribution of this paper is the proposed REWB algorithm which ensures that over a directed communication network, both the normal agents as well as the agents under attack asymptotically estimate the actual value of the unknown parameter in the presence of a sensor spoofing attack by the adversaries. 

Technically, the contributions we make in this paper can be summarized as follows: 
\begin{itemize} 
\item We propose a novel REWB algorithm that estimates an unknown time-varying parameter with a decaying bound on its variations in the presence of sensor spoofing attacks by simultaneously balancing the unbalanced directed communication network (Algorithm \ref{alg:rewb}).
The REWB algorithm brings together the weight-balancing and consensus+innovation approaches over relative time-scales to achieve this.
\item We show that the proposed REWB algorithm ensures convergence of each agent, both normal as well as those under attack, to the actual value of the unknown parameter provided less than half of the total agents are under attack at any given time (Theorem \ref{thm}). 
\item As an intermediate result, we provide an explicit rate of convergence of the Laplacian of an unbalanced weighted digraph to the Laplacian of the associated balanced digraph (Lemma \ref{lem:Lt_Linf}). 
\end{itemize}

\textit{Notations.} $\R$ denotes the set of \emph{real} numbers, and $\R^N$ represents the $N$-dimensional Euclidean space. For any set $\mathcal{S}$, the \textit{cardinality} of the set is denoted by $|\mathcal{S}|$. $\one \coloneqq (1,1,\hdots,1)$ and $\zro \coloneqq (0,0,\hdots,0)$, of appropriate dimensions. For a real-valued vector $v$, $v^T$ denotes the \emph{transpose} of the vector, $||v||$ denotes its $l_2$-norm and $||v||_\infty$ denotes its $\infty$-norm. Similarly for a real-valued matrix $M$, $M^T$ denotes the \emph{transpose} of the matrix, and $||M||$ denotes its \emph{spectral norm}. Among the \emph{eigenvalues} of $M$, $\lambda_2(M)$ represents the \emph{second lowest} eigenvalue of $M$ in magnitude, while $\lambda_{\text{max}}(M)$ denotes its largest eigenvalue in magnitude. For a real-valued vector $v$, diag$(v)$ represents a diagonal matrix with $v$ as the main diagonal.  

The rest of the paper is organised as follows. Section-\ref{sec:PF} discusses the details of the problem such as the inter-agent communication network, the threat model of the adversaries and the concept of resilience index. Section-\ref{sec:Rslt} starts with the development of the REWB algorithm using the weight-balancing approach, followed by the details of the algorithm, finally leading to our main result. Some numerical simulations are presented in Section-\ref{sec:SimRes} to validate the performance of the REWB algorithm. Finally the conclusions are presented in Section-\ref{sec:Con}.

\section{Problem Formulation} \label{sec:PF}

\subsection{System Model} \label{sec:PF_sysMdl}
Consider a system of $N$ agents where each agent is equipped with sensing, computing and communication capabilities - it can record measurements using its sensor, can perform computations using its own data and the information received from its neighbouring agents, and can also share its data with the neighbours. The aim of each agent is to estimate an unknown parameter $\ts{t} \in \R^M$ in a distributed manner even while some agents' sensor measurements are corrupted by adversaries. The precise model of sensor measurement corruption will be described shortly.

The communication among the agents is modelled as a directed graph $\Gamma = (\node,\edge)$, where the vertex set $\node = \{ 1, 2, \hdots, N \}$ represents the set of $N$ agents. The set of directed edges $\edge \subset \node \times \node $ represents the information exchange links between the agents, where $(i,j) \in \mathcal{E}$ if agent $j$ can send information to agent $i$. A \emph{directed path} from $i$ to $j$ is the sequence of directed edges $(i,i_1),(i_1,i_2),\hdots,(i_p,j)$. The set of \textit{in-neighbours} of agent $j$ is defined as $\mathcal{N}_j = \{ i \in \node : (j,i) \in \edge \}$, and the corresponding \textit{in-degree} is denoted as $d_j^{\textrm{in}} = |\mathcal{N}_j|$. The set of \textit{out-neighbours} of agent-$j$ is defined as $\mathcal{O}_j = \{ i \in \mathcal{V} : (i,j) \in \mathcal{E} \}$, and the corresponding \textit{out-degree} is denoted as $d_j^{\text{out}} = |\mathcal{O}_j|$. A corresponding diagonal matrix is defined as $D^{\text{out}} = \text{diag}\bbc{d_1^{\text{out}}, \hdots, d_N^{\text{out}}}$. The \emph{adjacency matrix}, $A$ is a square matrix of size $N \times N$ defined as $A = [a_{ij}]$ where $a_{ij} = 1$ if $(i,j) \in \edge$, and $a_{ij} = 0$ otherwise. The \emph{Laplacian}, $L$ is defined as $L \coloneqq \dout - A$. 

\begin{definition} [\textbf{Strongly Connected Graph }]
A directed graph is said to be \emph{strongly connected} if there exists a directed path between every pair of vertices in the graph. 
\end{definition}

The flow of information is such that each agent $i$ is able to receive information from its in-neighbours $(\mathcal{N}_i)$, and send its own data to its out-neighbours $(\mathcal{O}_i)$. So the information about any agent $i$ can be received by another agent $j$ either directly if a directed communication link exists between them, or indirectly via intermediate agent(s) provided the corresponding directed path exists. In order to ensure that the information about every agent $i$ reaches every other agent $j, (i \neq j; i,j \in \node)$, we introduce the following assumption.

\begin{assumption} \label{asmp:stCnctvty}
The directed graph $\Gamma$ is \textit{Strongly Connected}. 
\end{assumption}

Now we proceed to model the effect of the adversaries, which attack the agents with a motive to disrupt the estimation process thus trying to prevent them from correctly estimating the value $\ts{t}$. At every time-step $t \geq 0$, the agents which are under attack by the adversaries are termed as the the set of \textit{Bad} (or \textit{affected}) agents, denoted as $\bd$. The remaining agents form the set of \textit{Good} (or \textit{normal}) agents, denoted as $\gd$. The set of bad agents can vary with time, and are also not known a-priori to the set of good agents. So for each $t \geq 0$, the set $\node$ is partitioned into $\gd$ and $\bd$. Thus $\gd \cup \bd = \node, \forall $ $t \geq 0$. 
The attack model of the adversary is sensor spoofing attacks. Here the adversary introduces spurious signals into the sensor readings non-invasively \cite{SenSpoof_TAC16}. The corruption of sensor readings remains undetected by commonly used filters \cite{SenSpoof_ICCD20}. So even after nullifying the noise in sensor readings, the effect of the spoofing attack would still percolate into the measurements available to the agent. The agents use these sensor measurements to estimate the unknown parameter. In order to specifically highlight the effect of the adversary, we consider the sensor measurements available to the agents to be free of the effect of any measurement noise. The sensor measurements available to the agents under attack are arbitrary values manipulated by the adversary. Accordingly, we model the sensor measurements recorded by the agents as 
\begin{equation} \label{eq:par_yi} 
    \begin{split}
        \text{for all } i \in \gd &\text{ , } y_i(t) = \ts{t} \\
        \text{for all } i \in \bd &\text{ , } y_i(t) = \ts{t} + \zeta_i(t)
    \end{split}
\end{equation}
where $\zeta_i(t) \in \R^M$ is a vector of arbitrary values reflecting the effect of the adversaries. In the above model there is no boundedness assumption or stochastic approximation considered for $\zeta_i(t)$. This preserves the arbitrary nature of the data being manipulated by the adversary. So, if $|\bd| = 0$ $\forall t \geq 0$, then $y_i(t) = \ts{t}$ $\forall i \in \node$, and the estimation problem would be trivial as the sensor measurements directly provide the correct value of the parameter. Here we are interested in the non-trivial case where there exists some $t \geq 0$ such that $|\bd| \neq 0$. This means some of the sensor measurements would be corrupted as $y_i(t) = \ts{t} + \zeta_i(t)$ $\forall i \in \bd$. So each agent needs to perform some additional computations in order to estimate the true value of $\ts{t}$ in a distributed manner. It should be noted that under this threat model, the bad agents are still able to perform  their computations as per design as well as communicate with their neighbours.

The unknown time-varying parameter that is to be estimated is some physical quantity which can be measured by a sensor. So we can safely assume its Euclidean norm to be bounded. Moreover, we also assume that the variations in the unknown parameter asymptotically decay with time.

\begin{assumption} \label{asmp:thtaBnd}
The Euclidean norm of the unknown vector quantity that is to be estimated lies within an upper bound known to each agent : 
\begin{equation} \label{eq:thtaBnd}
     \nrm{\ts{t}}  \leq \Theta
\end{equation}
Also, the Euclidean norm of the variation in the unknown vector quantity has a decaying bound : 
\begin{equation} \label{eq:thtaVarBnd}
     \nrm{\ts{t+1} - \ts{t} } \leq 1 / (1+t)^{\theta_1}  
\end{equation}
\end{assumption}
As a consequence of \eqref{eq:thtaVarBnd}, we have the time varying parameter $\ts{t}$ eventually converging to some constant value $\hat{\theta}$. Specifically, $\ts{t} \to \hat{\theta} \text{ as } t \to \infty$.

\textit{Remark :} The above assumption focuses on a particular subset of dynamic parameter estimation. Note that this is a modest extension from the \textit{static} parameter estimation case. 

Now to estimate $\ts{t}$ in a distributed manner, for all $t \geq 0$ each agent $i$ maintains its own estimate of $\ts{t}$ denoted by $x_i(t) \in \R^M$, also referred to as the \emph{state} of agent $i$. In order to update the state, each agent $i$ follows the discrete-time single integrator dynamics :
\begin{equation} \label{eq:SIDyn}
    x_i(t+1) = x_i(t) + u_i(t) , t \geq 0
\end{equation}
So at every time step, each agent $i$ performs the following steps in the given sequence : 
\begin{enumerate}
    \item[$S1$ -] broadcasts its own estimate $x_i(t)$ to its out-neighbours $\mathcal{O}_i$
    \item[$S2$ -] receives the estimates from its corresponding in-neighbours : $x_j(t)$, $j \in \mathcal{N}_i$ 
    \item[$S3$ -] collects sensor measurement of $\ts{t}$ : $y_i(t)$
    \item[$S4$ -] updates its own estimate following \eqref{eq:SIDyn}, where $u_i(t) = f(y_i(t), \{ x_j(t), j \in \mathcal{N}_i \})$, and $f$ is defined later in Section \ref{sec:Rslt_Algo}.
\end{enumerate} 

In a distributed estimation problem with $x_i(t)$ as the state of agent $i$ and $\ts{t}$ as the parameter of interest, the aim is to achieve  
\begin{equation} \label{eq:Aim}
    x_i(t) \longrightarrow \ts{t} \text{ as } \tti \text{ , for all } i \in \node
\end{equation} 
For the \emph{resilient} estimation problem considered here, the additional challenge is to achieve \eqref{eq:Aim} even in the presence of adversaries attacking some of the agents. In order to quantify how resilient an algorithm is to the adversarial attacks, we use a measure called the \textit{Resilience Index} \cite{KarTAC19}. The resilience index $(s)$ is an upper bound on the fraction of agents which are under attack by the adversaries at any time-step $t$. So, $s \geq \frac{|\bd|}{N}$ for all $t \geq 0$, $s \in \R$. Thus $s = 0$ would indicate the trivial case where bad agents are totally absent. Having $s=1$ allows for the possibility of all the agents being under attack at any time-step $t$.

In the sequel, we initially proceed to design an algorithm which provides us with a suitable value of $u_i(t)$ $\forall t \geq 0$, for all  $i \in \node$ introduced in \eqref{eq:SIDyn}. Then we present our main result on how the newly designed algorithm, under the assumptions made so far, achieves \eqref{eq:Aim}. 

\section{Results} \label{sec:Rslt} 
The aim of each agent in the multi-agent system under consideration, is to estimate an unknown static parameter in a distributed manner, as given in \eqref{eq:Aim}. 
The technique used for the distributed estimation of $\ts{t}$ is based on the consensus+innovations approach \cite{KarSPMag13}. Based on this approach we proceed to design an algorithm such that the desired objective, $x_i(t) \longrightarrow \ts{t} \text{ as } \tti \text{ , for all } i \in \node$, is achieved through fulfilling the following two smaller goals simultaneously as $\tti$: 
\begin{enumerate}
    \item[$G1$ :] the state of each agent, $x_i(t)$, approaches the average of the states of all agents, $\Bar{x}(t) \coloneqq (1/N) \sum_{i=1}^N x_i(t) $ 
    \item[$G2$ :] $\Bar{x}(t)$ approaches the unknown value to be estimated, $\ts{t}$.
\end{enumerate} 

\subsection{Weight Balancing} \label{sec:Rslt_WB} 

In a Multi-Agent System (MAS), the communication network is usually modelled as a graph, with the nodes of the graph representing the agents and the edges between the nodes representing the corresponding communication links between the agents. When the flow of information between agents is \emph{bi-directional}, the model used is an \emph{undirected} graph. On the other hand, when the flow of information between agents is \emph{unidirectional}, a \emph{directed} graph (or \emph{digraph}) is required to model it. The directed edges of the digraph represent the unidirectional communication links while preserving their direction of information flow. A \emph{weighted} graph has each of its edges assigned a real or integer value, referred to as \emph{edge weights}. Unless specifically mentioned, the edge weights are taken as unity.

In case of an undirected graph, the sum of edge weights of the incoming edges is equal to the sum of the edge weights of the outgoing ones. But this balance in edge weights does not necessarily hold true in the case of a digraph. To overcome this imbalance, we need to find a suitable \emph{weight vector} $w \in \R^N$, where each outgoing edge of agent $i$ is assigned the weight $w_i$. These weights are said to \emph{balance the graph} if $w_i d_i^{\text{out}} = \sum_{j \in \mathcal{N}_i} w_j$. The notion of a balanced graph is formally defined below. 
\begin{definition} [\textbf{Balanced Graph}] \label{def:BalWts}  
A graph $\Gamma$ of $N$ nodes is said to be \emph{balanced} if there exists $w \in \R^N$ such that  
\begin{equation} \label{eq:BalWts}
    L \one = \zro \text{ , } \one^T L = \zro^T
\end{equation}
where $L \coloneqq (\dout-A) \text{diag} (w)$
\end{definition}

The weights $\bbc{w_1,w_2,\hdots,w_N}$ which balance a given digraph are called the \emph{balancing weights} of the corresponding digraph \cite{BalWts}. Note that an undirected graph is inherently balanced with $w=\one$ as the vector of balancing weights. On the other hand, for a strongly connected digraph the vector of balancing weights is non-trivial. Note that this vector of balancing weights is also unique to the given digraph, up to scaling \cite{BalWts}. For example, consider the strongly connected digraph shown in Fig.\ref{fig:digrph}. For this digraph the vector of balancing weights is $w=[0.5,1.5,1]^T$, which is non-trivial and unique up to scaling.

\begin{figure}
   \centerline{\includegraphics[width=0.15\textwidth]{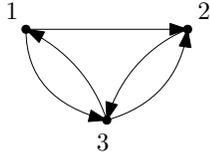}}
    \caption{A directed graph of 3 nodes}
    \label{fig:digrph}
\end{figure}

The SIU algorithm proposed in \cite{KarTAC19} for resilient estimation does not work in general for directed graphs. This is later illustrated through a numerical example in Fig.\ref{fig:states} in Section \ref{sec:SimRes}. We propose to use the idea of balancing weights described above to achieve resilient estimation over digraphs. For the directed graph $\Gamma$, we use the following update rule, proposed in \cite{BalWts}, to iteratively compute a set of balancing weights. Let $w_i(t) \in \R$ denote the weight at node $i$ at time-step $t$. The initial set of weights assigned to the agents satisfy : $w_i(0) \leq (1/\degOutMax)^{2\Phi+1}$, where $\degOutMax$ represents the maximum out-degree and $\Phi$ represents the diameter of the concerned digraph \cite{BalWts}. Then for all $t \geq 0$, 
\begin{equation} \label{eq:Updt_Wi} 
    w_i(t+1) = \frac{1}{2} w_i(t) + \frac{1}{d_i^{\text{out}}} \Bbc{\sum_{j \in \nei_i} \frac{1}{2} w_j(t) }.
\end{equation}
Let $w(t) = \bbc{w_1(t),w_2(t),\hdots,w_N(t)}$ represent the vector of node-weights at time-step $t$. Then the corresponding vector representation of \eqref{eq:Updt_Wi} is given by :
\begin{equation} \label{eq:Updt_W} 
    w(t+1) = P w(t)
\end{equation} 
where $P \coloneqq 0.5 \bbc{I + (D^{\text{out}})^{-1} A}$.
So for the limiting case, $\lim_{\tti} w(t) = \lim_{\tti} P^t w(0)$. 
Now from \emph{Lemma 1} of \cite{BalWts}, we know that $\lim_{\tti} P^t$ exists, and that the sequence $\{ w(t) \} _{t \geq 0}$ converges to the vector of balancing weights. So we define here the vector of weights which balances the digraph $\Gamma$ as 
\begin{equation} \label{eq:par_wInf}
    w^\infty \coloneqq \lim_{\tti} w(t) = \lim_{\tti} P^t w(0)
\end{equation}
The time-varying weighted Laplacian matrix is represented as
\begin{equation} \label{eq:par_Lt}
        L(t) = \big( D^{\text{out}} - A \big) W(t) , \text{ where } W(t) = \text{diag}\bbc{w(t)}
\end{equation} 
Then the Laplacian matrix for the limiting case can be defined using the result from \eqref{eq:par_wInf} in \eqref{eq:par_Lt} as
\begin{equation} \label{eq:par_LInf}
     \linf \coloneqq (\dout - A)\winf , \text{ where } \winf = \text{diag}\{w^\infty\}
\end{equation} 
Now as $w^\infty$ balances the digraph, $\linf$ satisfies the desired balancing condition expressed in \eqref{eq:BalWts}. By definition of $L(t)$ we have $\one^T L(t) = \zro^T$ for all $t \geq 0$. So to arrive at the desired balanced graph condition, we need $L(t) \one = \zro $ which is eventually achieved with $L(t)$ converging to $\linf$ as $\tti$. Next we state a lemma which provides an explicit rate for this convergence and additionally provides the rate of decay of $ L(t) \one $ to $\zro$.

\begin{lemma} \label{lem:Lt_Linf}
    Given $L(t) = \big( \dout - A \big) W(t)$ and $\linf = (\dout - A)\winf$, there exists constants $C > 0$ and $\eta \in (0,1)$, such that $\nrm{L(t) - \dout} \leq C \eta^t$ , $\nrm{L(t) \one} \leq C \eta^t$ for all $t \geq 0$.
\end{lemma}

The proof of Lemma \ref{lem:Lt_Linf} is given in Appendix \ref{sec:Apdx_Lem1}.
\subsection{Algorithm} \label{sec:Rslt_Algo} 
Now we introduce our algorithm, \emph{Resilient Estimation through Weight Balancing} (REWB). It consists of two main update steps : one for the state of the agents, and the other for the node weights.
\\ The updates performed by agent $i$ at time-step $t$ are :
\begin{enumerate}
    \item[i)] Updating the estimate 
    \begin{equation} \label{eq:Updt_Xi} 
\begin{split}
    x_i(t+1) = \bbc{1 &- \beta(t) w_i(t) d_i^{\text{out}}} x_i(t) \\ 
     &+ \beta(t) \Bbc{\sum_{j \in \nei_i} w_j(t) x_j(t)} \\ 
     &+ \alpha(t) k_i(t) \bbc{y_i(t) - x_i(t)}
\end{split}
\end{equation} 

    \item [ii)] Updating the weight 
    \begin{equation}
    w_i(t+1) = \frac{1}{2} w_i(t) + \frac{1}{d_i^{\text{out}}} \Bbc{\sum_{j \in \nei_i} \frac{1}{2} w_j(t) }
\end{equation}
\end{enumerate}

The update law \eqref{eq:Updt_Xi}, used by agents to update their estimate of $\theta^*(t)$, is based upon the consensus+innovation approach. The first two terms, dealing with the agent's own and neighbours' estimates and the corresponding node-weights, constitute the \emph{consensus} part of the update law. The third term, involving the measurements $y_i(t)$ and a scaling factor $k_i(t)$, constitute the \emph{innovation} part. These two parts working simultaneously through the same update law help in achieving the smaller goals $G1$ and $G2$ mentioned before. The above update law uses step-size parameters $\beta(t)$ and $\alpha(t)$ to assign proper weightage to its consensus and innovation parts respectively. The parameters are defined as :   
\begin{equation}\label{eq:AlpBet}
    \alpha(t) = \frac{\alpha_0}{(1+t)^{\alpha_1}} , \beta(t) = \frac{\beta_0}{(1+t)^{\beta_1}} 
\end{equation} 
where $0 < \alpha_0 \leq 1/(1 - 2s)$ , $0 < \beta_0 < \psi$ , $0 < \beta_1 < \alpha_1 < \theta_1$. The constant $\psi$ is defined as $\psi \coloneqq 2/\bbc{N \degInMax (\degInMax + \degOutMax)}$. Note that $\beta_1 < \alpha_1$ implies that, in the state update law \eqref{eq:Updt_Xi}, the weight of the innovation term decays faster than the weight of the consensus term.

The scaling factor, $k_i(t)$, is used in the innovation part in order to ensure that the effect of the adversaries on the state of an agent always remains bounded. 
\begin{equation} \label{eq:ki}
    k_{i}(t) \coloneqq 
        \begin{cases} 
         1 &, \text{ if } \nrm {y_i(t) - x_i(t)} \leq \gamma(t) \\
         \frac{\gamma(t)}{\nrm{y_i(t) - x_i(t)}} &, \text{ otherwise }
        \end{cases}
\end{equation} 
where $\gamma(t)$ is the output of a dynamical system defined as
\begin{equation} \label{eq:Gam}
    \gamma(t) \coloneqq \gmo{t} + \gmt{t}
\end{equation} 
The dynamics of $\gmo{t}$ and $\gmt{t}$ are defined as 
\begin{align}
    \gmo{t+1} &\coloneqq \bbc{1 - c_1 \mu(t) + (1+\sqrt{N}) \alpha(t)} \gmo{t} \nonumber \\ 
    & \label{eq:Gam1} \hspace{10 ex} + (1 + \sqrt{N}) \alpha(t) \gmt{t} + c_2 \eta^t \\
     \gmt{t+1} &\coloneqq \alpha(t) \gmo{t} + \bbc{1 - \alpha(t) (1-2s) } \gmt{t} \nonumber \\ 
    & \label{eq:Gam2} \hspace{20 ex} + 1/(1+t)^{\theta_1} 
\end{align}
where, $\mu(t) = \frac{\mu_0}{(t+1)^{\mu_1}}$, $\mu_0 > 0$, $\beta_1 < \mu_1 < \alpha_1 $, $c_1 > 0$, $c_2 > 0$, $0 < \eta < 1$. The above time-varying system in two variables plays a crucial role in proving our main result. 
From the definition of $k_i(t)$ in \eqref{eq:ki}, a corresponding diagonal matrix is defined as 
\begin{equation} \label{eq:par_Kt}
    K(t) \coloneqq \text{diag}\bbc{k_1(t), k_2(t), \hdots, k_N(t)}
\end{equation} 
Let $x(t) = \bbc{x_1^T(t), x_2^T(t), \hdots, x_N^T(t)} \in \R^{N \times M} $ represent the matrix whose rows are the state vectors of the agents at time-step $t$. Also let $y(t) = \bbc{y_1^T(t), y_2^T(t), \hdots, y_N^T(t)} \in \R^{N \times M}$ represent the matrix whose rows are the sensor measurements of the agents at time-step $t$. Now we summarise our REWB algorithm as follows :  
\begin{algorithm} 
\caption{ REWB} 
\label{alg:rewb}
 \textbf{Given} : Graph $\Gamma$, $\Theta \geq \nrm{\ts{t}} $, Resilience index $s$, and $\theta_1$
\\ \textbf{Initialize} : $0 < \alpha_0 \leq 1/(1 - 2s)$, $0 < \beta_0 < \psi$, $x(0)=0$, $\mu_0 < (\lambda_m - \beta_0 \lambda_M) \beta_0 / (2 c_1)$, $\gmo{0}=0, \gmt{0}=\Theta$ , $w_i(0) \leq \bbc{ \frac{1}{\degOutMax}}^{2\Phi+1}$
\\ \textbf{Choose} : $0 < \beta_1 < \mu_1 < \alpha_1 < \theta_1$
\\ \textbf{for} $t = 0,1,\hdots $ \textbf{do}
\begin{itemize}
    \item \textbf{record} $y(t)$ 
    \item \textbf{exchange} $x(t)$ among neighbouring agents 
    \item \textbf{update} $x(t)$ :
    \\ $x(t+1) = \bbc{I - \beta(t) L(t)} x(t) + \alpha(t) K(t) \bbc{y(t) - x(t)}$ 
    \item \textbf{update} $w(t)$ :
    \\ $w(t+1) = P w(t)$ 
    \item \textbf{update} $\gamma(t)$ : using equations \eqref{eq:Gam}, \eqref{eq:Gam1} \& \eqref{eq:Gam2} 
\end{itemize}
\textbf{end for}
\end{algorithm}
\subsection{Main Result} \label{sec:Rslt_MR} 
The following theorem states our main result on resilient distributed estimation using the REWB algorithm.  
\begin{theorem} \label{thm}
Suppose Assumptions \ref{asmp:stCnctvty} and \ref{asmp:thtaBnd} hold, and the effect of the adversaries is modelled as in \eqref{eq:par_yi}. Then the REWB algorithm ensures that the state of every agent, $x_i(t)$ converges to $\ts{t}$, provided $s \in [0,\frac{1}{2})$. In particular, 
\begin{equation} \label{eq:thm}
    \lim_{\tti} (t+1)^{\delta_1} \nrm{x_i(t) - \ts{t}} = 0 \text{ , for all } i \in \node
\end{equation} 
where $0 \leq \delta_1 \leq \alpha_1 - \beta_1$
\end{theorem}
The proof of Theorem-\ref{thm} is given in Appendix-\ref{sec:Apdx_Thm}. Here we provide a remark on the above theorem. 

\begin{remark}
    From Theorem-\ref{thm} it can be inferred that as long as less than half the total number of agents are under attack by the adversaries, the REWB algorithm ensures that each agent correctly estimates $\ts{t}$. Also note that all the agents, even the bad agents, achieve consensus and estimate $\ts{t}$ in a distributed manner.
\end{remark}
\begin{remark}
    As noted in Lemma \ref{lem:Lt_Linf}, the dynamic weights $w(t)$ converge to the balancing weights at an exponential rate ($\eta^t$), whereas all time-varying signals in the dynamics of the state update rule converge at a polynomial rate ($\alpha(t) = \alpha_0/(1+t)^{\alpha_1}$ etc.). Thus, the weights converge faster, which are in turn used in the state update rule. This two time-scales approach facilitates convergence of the algorithm.
\end{remark} 

\section{Simulation Results} \label{sec:SimRes}
We evaluate the performance of our proposed REWB algorithm through numerical simulations. A random network, consisting of 100 agents with directed edges, is generated where each possible edge has a probability of $0.5$. It models the communication network among the agents. Each agent estimates a scalar time-varying parameter $\theta^*(t) = 25 + 1/(t+1)$ with $\Theta = 50$ and $ \theta_1 = 1$. The required algorithm parameters are chosen as : $\alpha_0 = 0.01 , \alpha_1 = 0.075 , \beta_0 = 0.01 , \beta_1 = 0.01 , \mu_0 = 0.025 , \mu_1 = 0.025 , c_1 = 75, c_2 = 75$ , and $\eta = 0.5$. The initial weights are chosen as $w_i(0)=0.1 $ $\forall i \in \node$.

The noise term $\zeta_i(t)$ models the effect of the adversaries on the sensor measurements of agent $i$. For each bad agent $i \in \bd$, at every time step, $\zeta_i(t)$ takes on a random value uniformly distributed between 0 and $-\Theta$. Note that the REWB algorithm works for any other range also. We select the base resilience index to be $s = 0.405$, and correspondingly choose $|\bd| = 40$. At first we consider two cases with respect to the set of agents under attack and observe the performance of the REWB algorithm. In Fig. \ref{fig:rNoiseA}, $\bd$ has a fixed set of agents, while in Fig. \ref{fig:rNoiseB}, $\bd$ is allowed to vary with time. Both the plots in Fig. \ref{fig:rNoise} show the error in estimation of $\theta^*(t)$ by the agents, given by $||x(t) - \theta^*(t) \one||$. From the proof of Theorem \ref{thm} we have $|x_i(t) - \theta^*(t)| \leq \gamma(t)$, $\forall i \in \node$, $\forall t \geq 0$. Then for a set of $N$ agents, we have $||x(t) - \theta^*(t) \one|| \leq \sqrt{N} \gamma(t)$. Fig. \ref{fig:rNoise} shows that, regardless of adversaries attacking a fixed or varying set of agents, the REWB algorithm ensures that the estimation error always remains bounded by $\sqrt{N} \gamma(t)$, and consequently dies down asymptotically. 

\begin{figure}
    \centering
    \begin{subfigure}[b]{0.475\columnwidth}
         \centering
         \includegraphics[width=\textwidth]{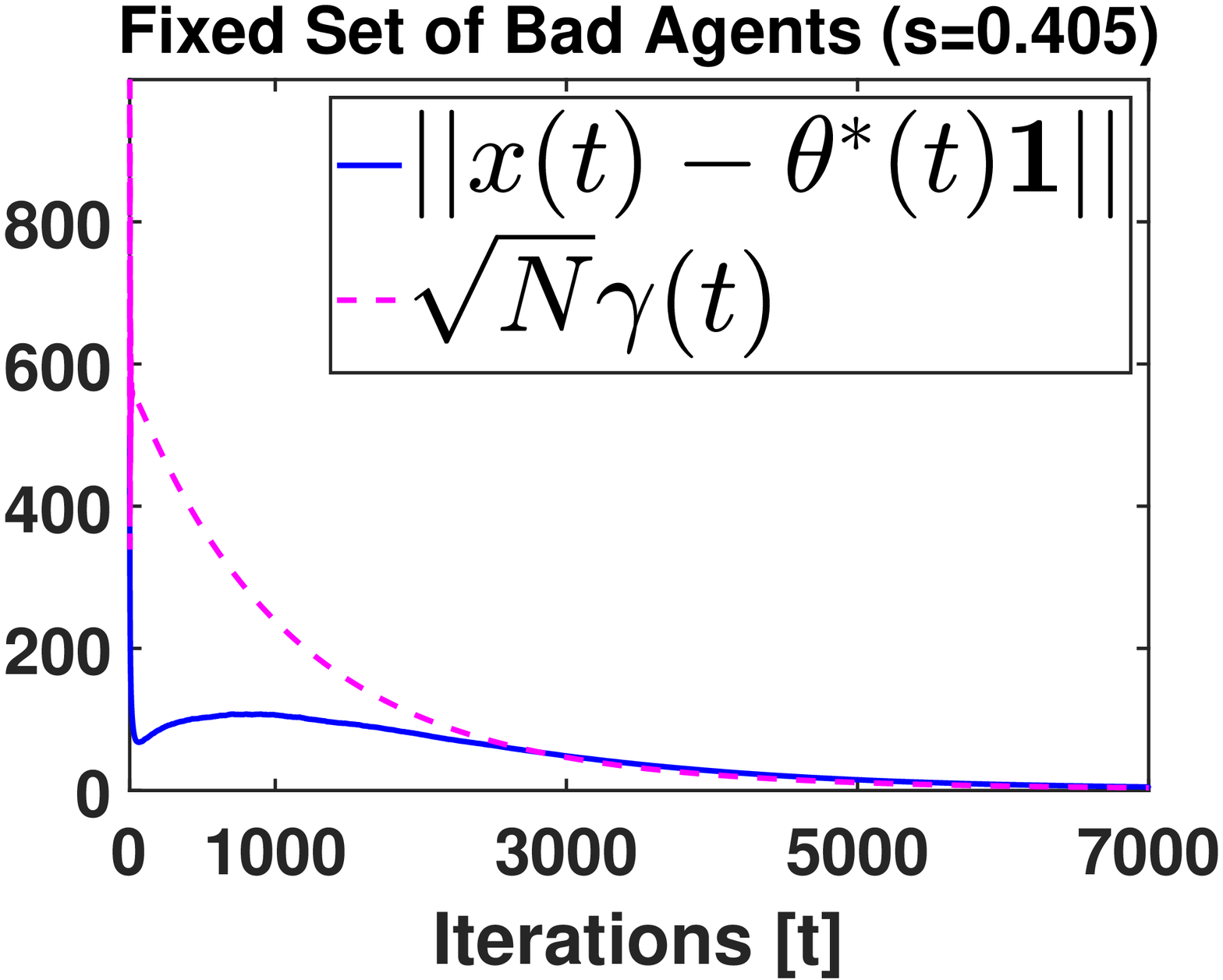}
         \caption{}
         \label{fig:rNoiseA}
     \end{subfigure}
    \hfill
    \begin{subfigure}[b]{0.475\columnwidth}
         \centering
         \includegraphics[width=\textwidth]{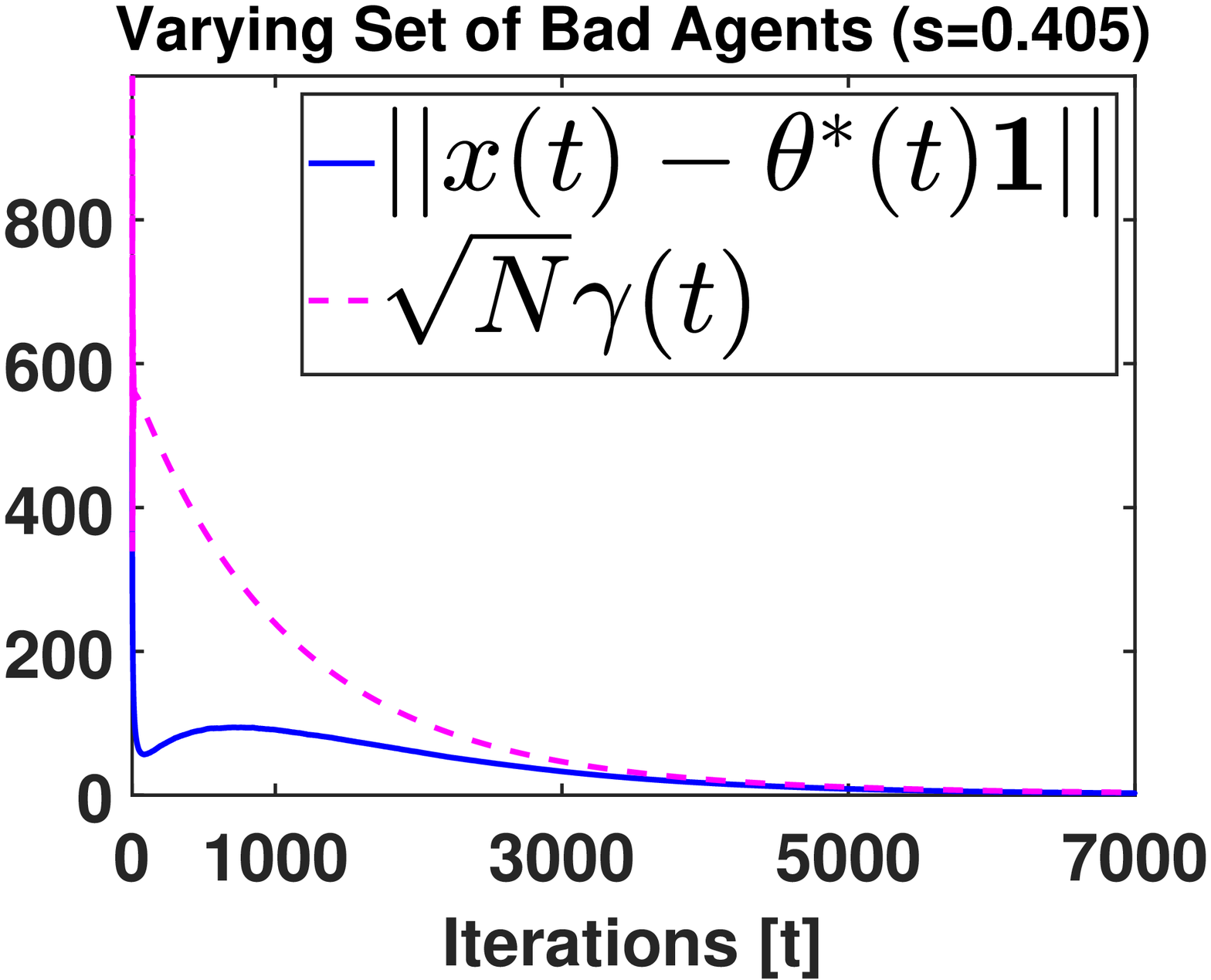}
         \caption{}
         \label{fig:rNoiseB}
     \end{subfigure}
    \caption{Performance of REWB with adversarial attacks on 40 agents where the set of bad agents is (a) fixed, and (b) variable.}
    \label{fig:rNoise}
\end{figure}

Next we use two different variations in operating conditions compared to the one used in Fig. \ref{fig:rNoiseA} and observe their effect in the performance of the REWB algorithm in Fig. \ref{fig:fxdSet}. For the plot in Fig. \ref{fig:fxdSetA}, the resilience index is decreased to $s=0.255$ and correspondingly we choose $|\bd|=25$. As can be observed the estimation error dies down much faster with a decrease in $s$. Next for the plot Fig. \ref{fig:fxdSetB}, we simulate an increase in the degree of manipulation done by the adversaries on the sensor measurements by increasing the noise level. We assign $\zeta_i(t) = 5 \Theta$ $\forall i \in \bd$, $\forall t \geq 0$. As is evident from Fig. \ref{fig:fxdSetB}, a high value of $\zeta_i(t)$ is also quite efficiently handled by the REWB algorithm, with the estimation error remaining bounded by $\sqrt{N} \gamma(t)$ at all times and eventually converging to 0. From Fig. \ref{fig:rNoise} and Fig. \ref{fig:fxdSet} it is evident that the REWB algorithm ensures that even the bad agents are able to eventually correctly estimate the true value of $\theta^*(t)$, along with the good agents. This is in accordance with the Remark stated in Section \ref{sec:Rslt_MR}.

\begin{figure}
    \centering
    \begin{subfigure}[b]{0.475\columnwidth}
         \centering
         \includegraphics[width=\textwidth]{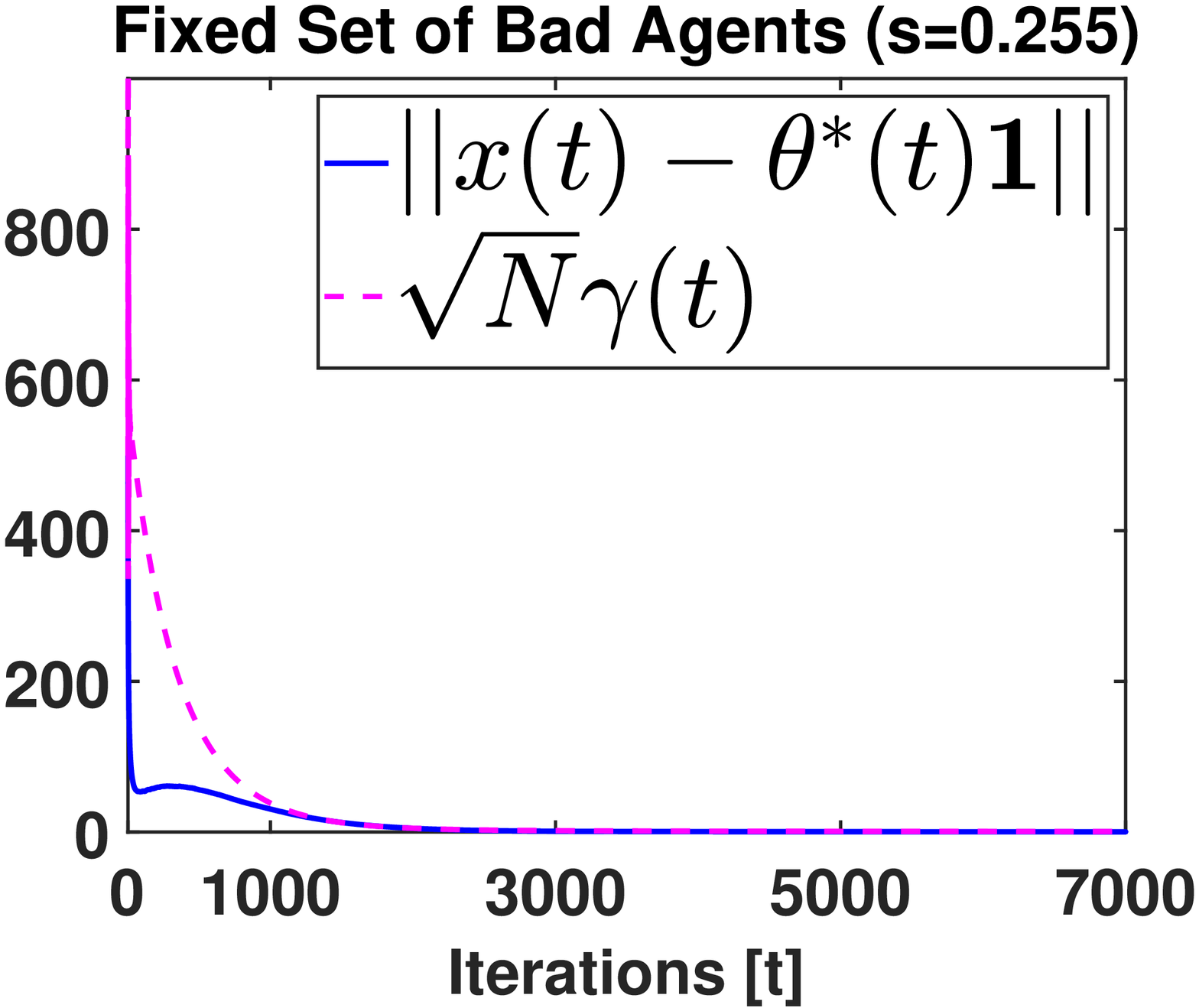}
         \caption{}
         \label{fig:fxdSetA}
     \end{subfigure}
    \hfill
    \begin{subfigure}[b]{0.475\columnwidth}
         \centering
         \includegraphics[width=\textwidth]{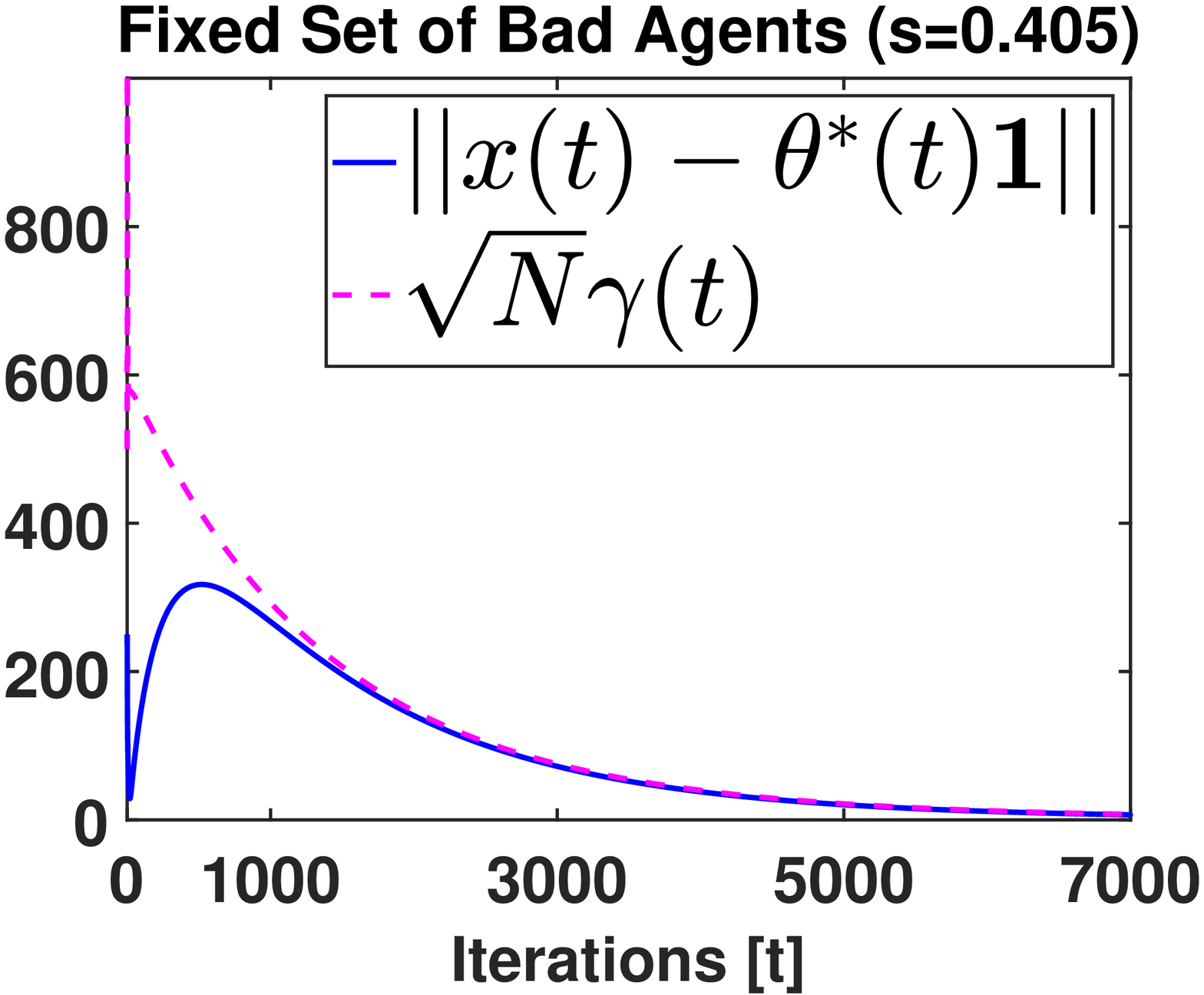} 
         \caption{}
         \label{fig:fxdSetB}
     \end{subfigure}
    \caption{Performance of REWB with (a) decrease in resilience index, and (b) increased manipulation in sensor measurements by adversaries.}
    \label{fig:fxdSet}
\end{figure}

In Section \ref{sec:Rslt_WB}, we mentioned that the SIU algorithm in \cite{KarTAC19} does not give convergence in general when applied over a directed network of agents. In Fig. \ref{fig:states}, we compare the performance of our REWB algorithm with the SIU algorithm in estimating the value of a scalar constant parameter $\theta^* \in \R$ over a directed network of 100 agents with $s=0.405$. The two plots on the left show how the states of the agents behave with time, while the two plots on the right show the net estimation error. Fig. \ref{fig:statesA} shows how on applying the SIU algorithm, the states of the agents diverge away from each other and never achieve consensus, leading to a constant estimation error. On the other hand, Fig. \ref{fig:statesB} shows how our REWB algorithm not only ensures the agents reach consensus but they also correctly estimate the value of $\theta^*$. This is made possible by the introduction of the weight balancing idea while designing the REWB algorithm. The dynamics of the time-varying weights ensure that the weighted graph eventually approaches a balanced condition, and thus consensus is achieved. 

\begin{figure}
   \centering
   \begin{subfigure}[b]{\columnwidth}
        \centering
        \begin{subfigure}[b]{0.475\columnwidth}
            \centering
            \includegraphics[width=\textwidth]{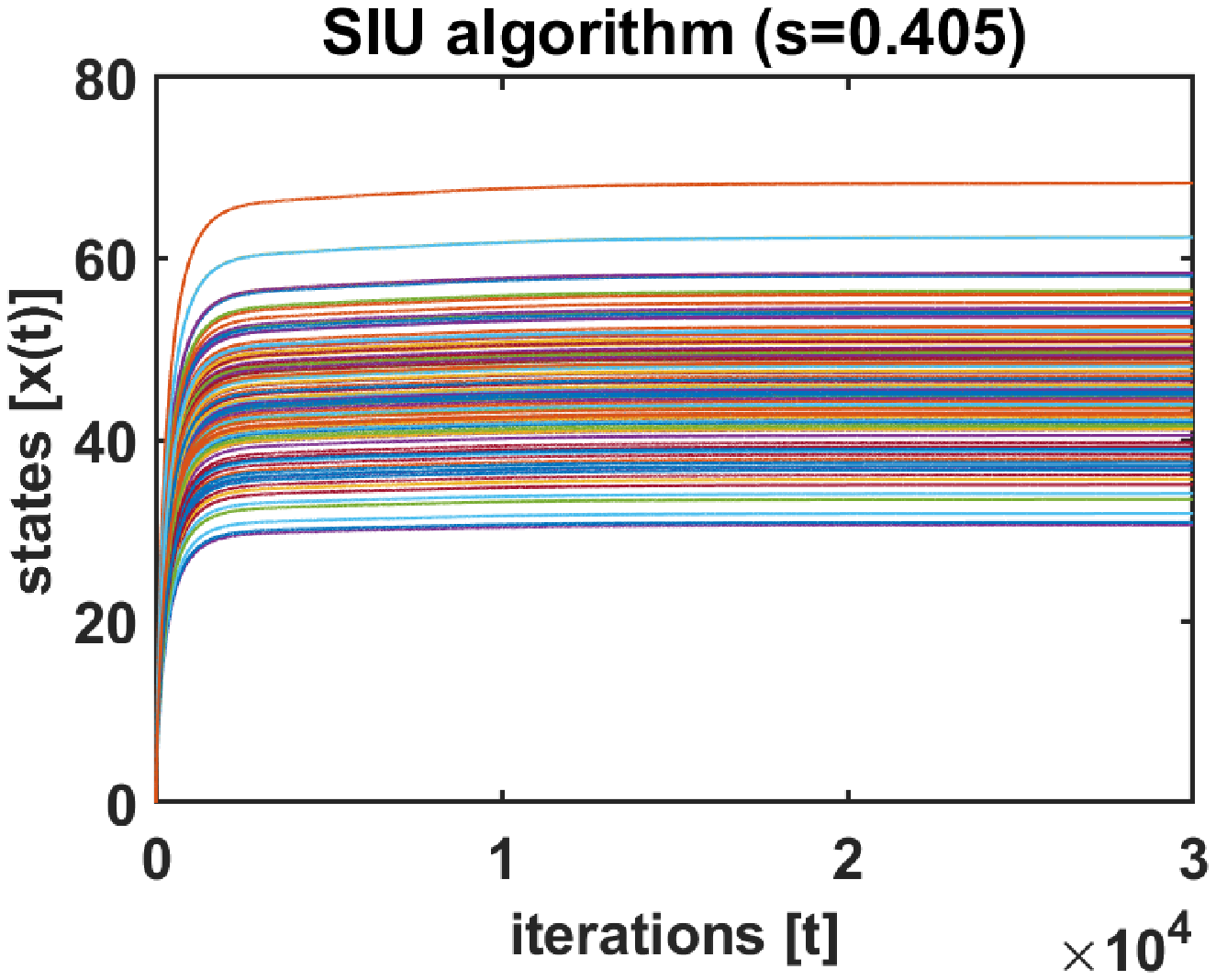}
        \end{subfigure}
        \hfill
        \begin{subfigure}[b]{0.475\columnwidth}
            \centering
            \includegraphics[width=\textwidth]{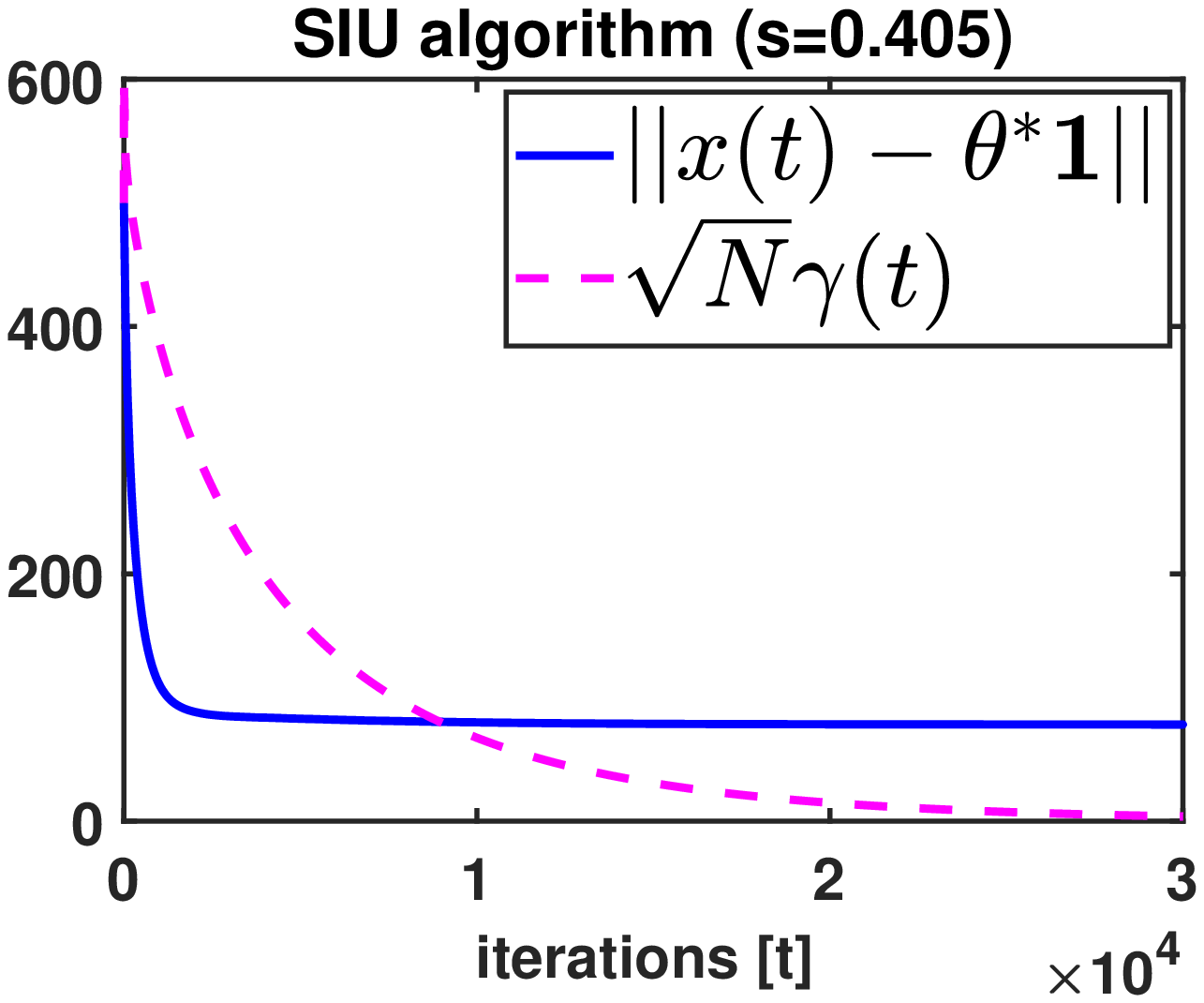}
        \end{subfigure}
        \caption{}
        \label{fig:statesA}
    \end{subfigure}
    \centering
    \begin{subfigure}[b]{\columnwidth}
        \centering
        \begin{subfigure}[b]{0.475\columnwidth}
            \centering
            \includegraphics[width=\textwidth]{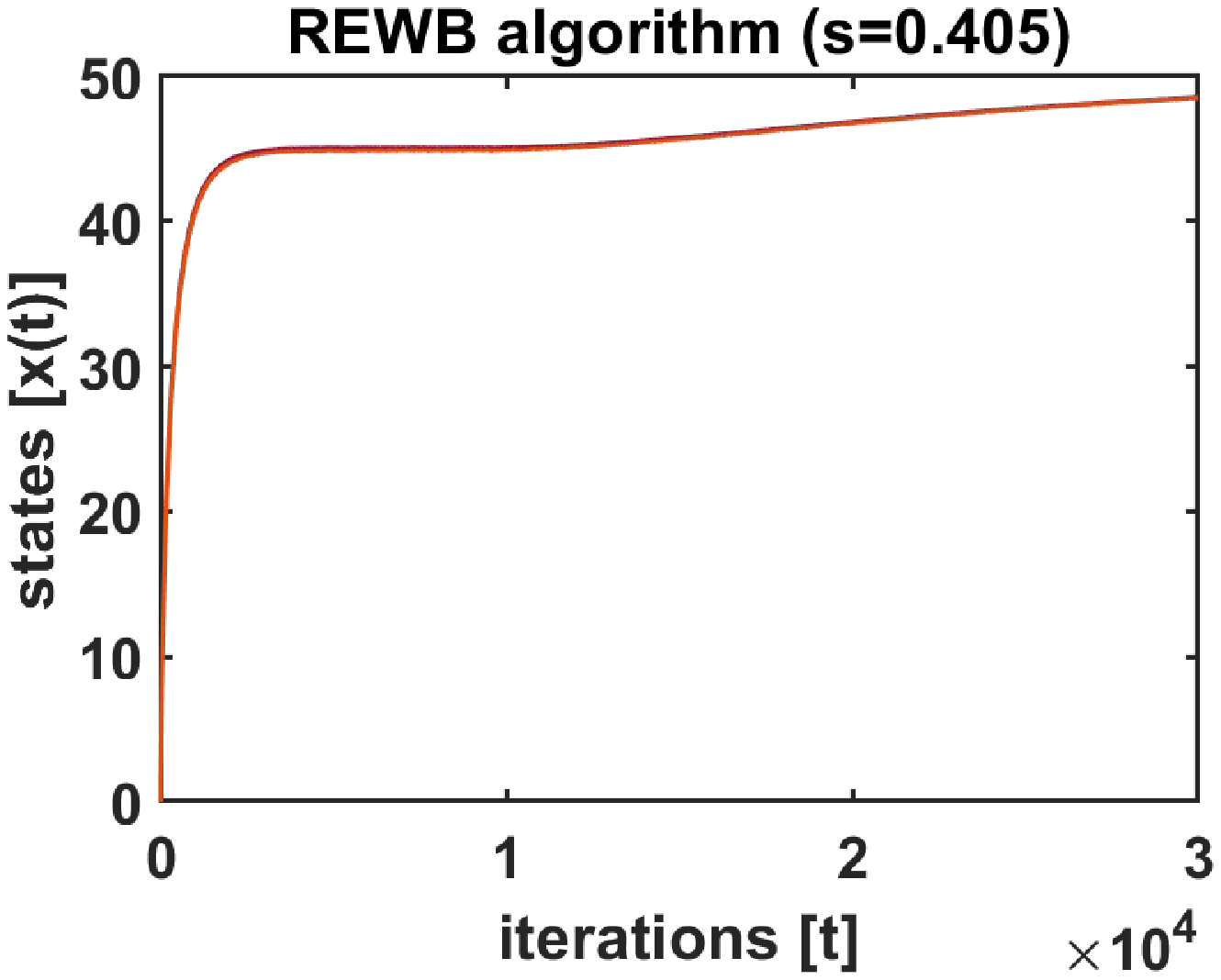}
        \end{subfigure}
        \hfill
        \begin{subfigure}[b]{0.475\columnwidth}
            \centering
            \includegraphics[width=\textwidth]{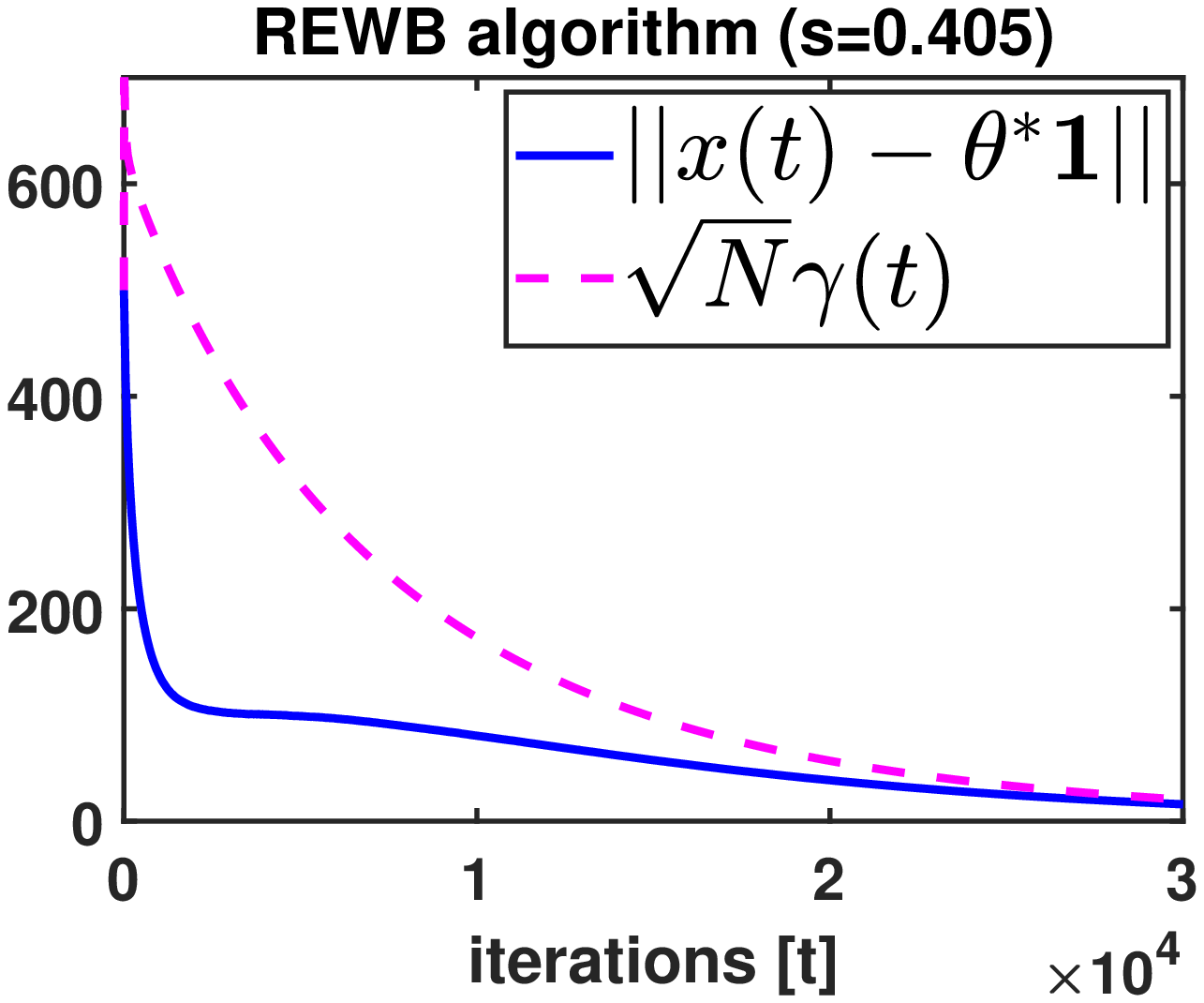}
        \end{subfigure}
        \caption{}
        \label{fig:statesB}
    \end{subfigure}
    \caption{Performance of (a) the SIU \cite{KarTAC19} algorithm and (b) the proposed REWB algorithm over directed graphs}
    \label{fig:states}
\end{figure}

\section{Conclusion} \label{sec:Con}
In this paper we propose the Resilient Estimation through Weight Balancing (REWB) algorithm. It is a distributed estimation algorithm designed to work for a network of sensor nodes with directed communication links. The REWB algorithm is resilient to sensor spoofing type adversarial attacks while estimating an unknown time-varying parameter with a decaying bound on its variations. It ensures that along with the unaffected agents, even the agents under attack estimate the true value of the parameter in a distributed manner. The proposed algorithm is developed based on the consensus+innovation approach and uses the weight-balancing idea to ensure consensus over directed graph. Through numerical simulations it is shown that the proposed algorithm accurately estimates an unknown parameter under different attack conditions, provided less than half of the agents are under adversarial attack at any given point of time. Future direction of work is to consider other models of adversarial attacks.

\useRomanappendicesfalse
\appendices 
\section{} \label{sec:Apdx_Lem1} 
\begin{proof}[Proof of Lemma \ref{lem:Lt_Linf}]
By definition, $P$ is a primitive matrix with spectral radius 1 \cite{BalWts}. Then by properties of primitive matrices :  $\lim_{\tti} P^t$ exists, and $\lim_{\tti} P^t = P_\infty = u v^T$, where $u, v$ are the right and left eigen-vectors of $P$ corresponding to eigen-value 1, and $v^T u = 1$. Then from \eqref{eq:Updt_Wi} and \eqref{eq:par_wInf} we have :
\begin{equation} \label{eq:lem1_1}
    w(t) - w^\infty = (P^t - P_\infty) w(0) 
\end{equation} 
Using the properties of $P$ discussed above we get $(P - P_\infty)^t = P^t - P_\infty$, for all $t \geq 1$.
Then from \eqref{eq:lem1_1} we have 
\begin{equation*}
     w(t) - w^\infty = (P^t - P_\infty) w(0) = (P - P_\infty)^t w(0)
\end{equation*}
So by Theorem 8.3 in \cite{Hesp}, there exists $c > 0, \eta < 1$ such that for all $t \geq 0$
\begin{equation} \label{eq:bnd_WtWinf}
    \nrm{w(t) - w^\infty} \leq c \eta^t \nrm{w(0)}
\end{equation} 
Using \eqref{eq:par_Lt} and \eqref{eq:par_LInf}, and applying the properties of sub-multiplicativity of spectral norm we have 
\begin{equation} \label{eq:LtLinf_WtWinf}
    \nrm{L(t) - \linf} \leq \nrm{\dout - A} \nrm{W(t) - \winf}
\end{equation} 
Now $W(t), \winf$ are diagonal matrices and for any diagonal matrix $M = \text{diag}(m_1,\hdots,m_N)$ we have $||M|| \leq ||m||_\infty \leq ||m||$. Applying this to \eqref{eq:LtLinf_WtWinf} and using the result from \eqref{eq:bnd_WtWinf} we get
\begin{equation} \label{eq:bnd_LtLinf}
    \nrm{L(t) - \linf} \leq C_L \eta^t
\end{equation}
where $C_L = c \nrm{\dout - A} \nrm{w(0)} > 0$. 
\\ From \eqref{eq:par_Lt} using the sub-multiplicativity property of the norm and the result from \eqref{eq:bnd_LtLinf} we get
\begin{equation} \label{eq:bnd_LtOne}
    \nrm{L(t) \one} = \nrm{(L(t) - \linf) \one} \leq \sqrt{N} C_L \eta^t
\end{equation}
By choosing $C \geq \sqrt{N} C_L$ and applying to \eqref{eq:bnd_LtLinf} and \eqref{eq:bnd_LtOne} we get : 
\begin{equation*}
    \nrm{L(t) - \linf} \leq C \eta^t \text{  ,  } \nrm{L(t) \one} \leq C \eta^t
\end{equation*}
\end{proof}

\section{} \label{sec:Apdx_IL} 
Here we introduce some \emph{intermediate lemmas} which will be useful in the proof of Theorem-\ref{thm}. At first, before proceeding to a time-varying system in two variables, we first analyse the dynamics of a scalar time-varying system. Consider a linear scalar time-varying system - 
\begin{equation} \label{eq:TVsys_Sclr}
        v_{t+1} = \big( 1 - r_1(t) \big) v_t + r_2(t)
\end{equation} 
where \begin{equation} \label{eq:TVsys_r1r2}
        r_1(t) = \frac{c_1}{(1+t)^{\delta_1}} , r_2(t) = \frac{c_2}{(1+t)^{\delta_2}}
    \end{equation} 
where $c_1, c_2, \delta_2$ are positive constants, and $0 \leq \delta_1 \leq 1$.

The following result is based upon the results introduced in Lemma 25 in \cite{KarIT12} and Lemma 3 in \cite{KarTAC19}. It provides a relation between $\delta_1$ and $\delta_2$ under which the dynamics of the scalar time-varying system in \eqref{eq:TVsys_Sclr} is bounded. It also gives the condition under which the system dynamics converges to zero, and the corresponding rate of convergence. 
\begin{proposition}
\label{lem:Kar_lem2}
 Consider the system given in \eqref{eq:TVsys_Sclr} where $r_1(t), r_2(t)$ is given by \eqref{eq:TVsys_r1r2}. Then if $\delta_1 = \delta_2$, there exists $B > 0$, such that for sufficiently large non-negative integers $j < t$ , 
\begin{equation*}
    0 \leq \sum_{k=j}^{t-1} \left( \prod_{l=k+1}^{t-1} \bbc{1 - r_1(l)} \right) r_2(k) \leq B 
\end{equation*} 
Moreover the constant $B$ can be chosen independently of $t,j$. 
Also, if $\delta_2 > \delta_1 $, then for arbitrary fixed $j$, 
\begin{equation*}
    \lim_{\tti} \sum_{k=j}^{t-1} \left( \prod_{l=k+1}^{t-1} \bbc{1 - r_1(l)} \right) r_2(k) = 0
\end{equation*}
and correspondingly $\lim_{\tti} (t+1)^{\delta_0} v_t = 0$
for all $0 \leq \delta_0 < \delta_2 - \delta_1$, and for all initial conditions $v_0$.
\end{proposition}
The following result provides the rate of convergence of a scalar system modified from \eqref{eq:TVsys_Sclr}.
\begin{proposition} [Lemma 4 in \cite{KarTAC19}] 
\label{lem:Kar_lem4}
Consider the scalar time-varying linear system : 
    \begin{equation}\label{lem1}
        v_{t+1} = \big( 1 - c_3r_2(t) + c_4r_1(t) \big) v_t + c_5r_1(t)
    \end{equation} 
    where $r_1(t), r_2(t)$ are given by : 
    \begin{equation*}
        r_1(t) = \frac{c_1}{(1+t)^{\delta_1}} , r_2(t) = \frac{c_2}{(1+t)^{\delta_2}}
    \end{equation*} 
    where $c_1,c_2,\hdots,c_5 > 0$, and $0 < \delta_2 < \delta_1 < 1$. 
    \\ The system in \eqref{lem1} satisfies $\lim_{\tti} (t+1)^{\delta_0} v_t = 0$
    for all $0 \leq \delta_0 < \delta_1 - \delta_2$, and for all initial conditions $v_0$.
\end{proposition} 

Now by using the results above we introduce the following lemma
which proves the convergence of $\gmo{t}$ and $\gmt{t}$ introduced in \eqref{eq:Gam1} and \eqref{eq:Gam2}.
\begin{lemma}   
\label{lem:convergence} 
The system in \eqref{eq:Gam1} and \eqref{eq:Gam2} satisfies 
    \begin{equation} \label{lem2_2shw1}
        \lim_{\tti} (t+1)^{\delta_0} \gmo{t} = 0 
    \end{equation} 
    \begin{equation} \label{lem2_2shw2}
        \lim_{\tti} (t+1)^{\delta_0} \gmt{t} = 0 
    \end{equation}
    where $0 \leq \delta_0 < \alpha_1 - \mu_1$
\end{lemma}

\begin{proof}
\textit{Step 1 :} As $\alpha(t), \mu(t)$ are decreasing in $t$, and $\alpha_1 > \mu_1$ , there exists a finite $T > 0$ such that for all $t > T$ 
\begin{equation} \label{lem2_1}
\begin{split}
    0 &\leq 1 - (1-2s) \alpha(t) \leq 1 \\ 
    0 &\leq 1 - c_1 \mu(t) + (1+\sqrt{N}) \alpha(t) \leq 1 
\end{split}
\end{equation}
From \eqref{eq:Gam2} we can express $\gmt{t}$ as  
\begin{multline}  \label{gmt_dyn} 
    \gmt{t} = \prod_{\tau=T}^{t-1} (1 - (1 - 2s) \alpha(\tau)) \gmt{T} + \\
     \sum_{\tau=T}^{t-1} \left( \prod_{j=\tau+1}^{t-1} (1 - (1 - 2s) \alpha(j)) \right) \big( \alpha(\tau) \gmo{\tau} + \frac{1}{(1+\tau)^{\theta_1}} \big) .
\end{multline}
Let $s(t) := \sum_{\tau=T}^{t-1} \left( \prod_{j=\tau+1}^{t-1} (1 - (1 - 2s) \alpha(j)) \right) \frac{1}{(1+\tau)^{\theta_1}} $.
Using the second part of Proposition \ref{lem:Kar_lem2}, we obtain : $s(t) \rightarrow 0$ as $t \rightarrow \infty$. Hence, $\exists T > 0$ such that $|s(t)| \leq \gamma_1(T)$ $\forall t \geq T$. Using this, along with \eqref{lem2_1}, in \eqref{gmt_dyn} provides
\begin{equation} \label{gmt_1}
    |\gmt{t}| \leq |\gmt{T}| + \sigma_1 \sup_{l \in [T,t]} |\gmo{l}| 
\end{equation}
for some constant $\sigma_1 > 0$.

\textit{Step 2 :} From \eqref{eq:Gam1} and \eqref{lem2_1} we have 
\begin{equation} \label{gmo_1}
\begin{split}
    |\gmo{t+1}| &\leq \bbc{1 - c_1 \mu(t) + (1+\sqrt{N}) \alpha(t)} \sup_{l \in [T,t]} |\gmo{l}| \\
     &+ (1 + \sqrt{N}) \alpha(t) |\gmt{t}| + c_2 \eta^t .
\end{split}
\end{equation}
Applying \eqref{gmt_1} we get 
\begin{equation*} 
\begin{split}
    \therefore |\gmo{t+1}| &\leq \bbc{1 - c_1 \mu(t) + \sigma_2 \alpha(t)} \sup_{l \in [T,t]} |\gmo{l}| \\
     &+ \sigma_3 \alpha(t) + c_2 \eta^t
\end{split}
\end{equation*}
where $\sigma_2 = (1 + \sqrt{N}) (1+ \sigma_1)$ and $\sigma_3 = (1 + \sqrt{N})  |\gmt{T}| $.
Now as $0 < \eta < 1$, there exists $\sigma_4 > 0$ such that for all $t > 0$
\begin{equation} \label{approx_1}
    \sigma_3 \alpha(t) + c_2 \eta^t < \sigma_4 \alpha(t)
\end{equation} 
\begin{equation*}
    \therefore |\gmo{t+1}| \leq \bbc{1 - c_1 \mu(t) + \sigma_2 \alpha(t)} \sup_{l \in [T,t]} |\gmo{l}| + \sigma_4 \alpha(t)
\end{equation*}
We define a new system - 
\begin{equation} \label{m_dyn}
    m(t+1) = \textrm{max}(m(t), \bbc{1 - c_1 \mu(t) + \sigma_2 \alpha(t)} m(t) + \sigma_4 \alpha(t))
\end{equation}
for all $t > T$ and initial condition $m(T) = \gmo{T}$. So by definition of $m(t)$ we have :
\begin{equation} \label{lem2_2}
    m(t) \geq \sup_{l \in [T,t]} |\gmo{l}| 
\end{equation}
We define another new system : 
\begin{equation} \label{mTld_dyn}
    \tilde{m}(t+1) =  \bbc{1 - c_1 \mu(t) +  \sigma_2 \alpha(t)} \tilde{m}(t) + \sigma_4 \alpha(t))
\end{equation}
for all $t > T$ and initial condition $\tilde{m}(T) = m(T) = \gmo{T}$.
By definition $\tilde{m}(T) \geq 0$. Also for $t > T$, from \eqref{lem2_1} we have $1 - c_1 \mu(t) + \sigma_2 \alpha(t) \geq 0$. Then $\tilde{m}(t) \geq 0$ for all $t \geq T$. 
Now using Proposition \ref{lem:Kar_lem2} and \eqref{mTld_dyn} we have 
\begin{equation*}
    \lim_{\tti} \tilde{m}(t) = 0
\end{equation*} 
\textit{Step 3 :} By virtue of $\tilde{m}(t)$ being a non-negative sequence which converges to 0, there exists a time $T_1 \geq T$ such that $\tilde{m}(T_1 + 1) \leq \tilde{m}(T_1)$. We choose the smallest value among all such possible $T_1 \geq T$. Then from the definition of $T_1$ we have $\tilde{m}(T) < \tilde{m}(T+1) < \hdots < \tilde{m}(T_1)$. So from \eqref{m_dyn}, $m(t) = \tilde{m}(t)$ for all $t \in [T,T_1]$.
\begin{equation} \label{lem2_3}
    \therefore m(t) \leq m(T_1) \text{ , for all } t \in [T,T_1]
\end{equation}
Also by definition of $T_1, m(t)$ we have $m(T_1 +1) = m(T_1)$.
\\ Let for all $t \geq T_1$
\begin{equation*}
    \pi(t) \coloneqq m(T_1) - \bbc{1 - c_1 \mu(t) + \sigma_2 \alpha(t)} m(T_1) - \sigma_4 \alpha(t)
\end{equation*}
By algebraic manipulation  
\begin{equation*}
    \pi(t) = \left( \frac{\sigma_5}{(t+1)^{\mu_1}} - \frac{\sigma_6}{(t+1)^{\alpha_1}}  \right) m(T_1)
\end{equation*} 
where $\sigma_5 = c_1 \mu_0 > 0$ , $\sigma_6 = \left( \sigma_2 + \frac{\sigma_4}{m(T_1)} \right) \alpha_0 > 0$. 
\\ Now $m(T_1) = 0$, and since $m(T_1 +1) = m(T_1)$ we have  
\begin{equation*}
    \pi(T_1) \geq 0 \iff T_1 \geq \left( \frac{\sigma_6}{\sigma_5} \right) ^{1/(\alpha_1 - \mu_1)} - 1
\end{equation*} 
So we have $\pi(t) \geq 0$ for all $t \geq T_1$. Then using \eqref{m_dyn} we have 
\begin{equation} \label{lem2_4}
    m(t) = m(T_1) \text{ , for all } t \geq T_1
\end{equation}
Now combining the results from \eqref{lem2_2}, \eqref{lem2_3} and \eqref{lem2_4} we get 
\begin{equation} \label{gmoMod}
   \sup_{t \geq 0} |\gmo{t}| < \infty
\end{equation}
\textit{Step 4 :} Let $\sup_{t \in [T,t]} |\gmo{t}| = B_1 < \infty$. Then from \eqref{gmt_1} we have 
\begin{equation} \label{lem2_7}
    \sup_{t \geq T} |\gmt{t}| \leq |\gmt{T}| + \sigma_1 B_1 < \infty 
\end{equation}
As $T < \infty$, we have 
\begin{equation} \label{lem2_8}
    \sup_{t \in [0,T]} |\gmt{t}| < \infty
\end{equation}
So combining \eqref{lem2_7} and \eqref{lem2_8} we have 
\begin{equation} \label{gmtMod}
   \sup_{t \geq 0} |\gmt{t}| < \infty
\end{equation}

\textit{Step 5 :} Let $\sup_{t \geq 0} |\gmt{t}| = B_2 < \infty$. 
Then for sufficiently large $t$, from \eqref{gmo_1} we have 
\begin{equation*}
\begin{split}
    |\gmo{t+1}| &\leq |1 - c_1 \mu(t) + (1+\sqrt{N}) \alpha(t)| |\gmo{t}| \\ 
     &+ (1 + \sqrt{N}) \alpha(t) B_2 + c_2 \eta^t
\end{split}
\end{equation*}
Now as $0 < \eta < 1$, there exists $C_\eta > 0$ and $T_\eta > 0$ such that for all $t > T_\eta$
\begin{equation} \label{approx_3}
    (1 + \sqrt{N}) B_2 \alpha(t) + c_2 \eta^t < C_\eta \alpha(t)
\end{equation} 
For a suitable choice of $C_\eta = \sigma_7$, \eqref{approx_3} holds for all $t > 0$. 
\begin{equation} \label{lem4_1}
    \therefore |\gmo{t+1}| \leq |1 - c_1 \mu(t) + (1+\sqrt{N}) \alpha(t)| |\gmo{t}| + \sigma_7 \alpha(t)
\end{equation}
As \eqref{lem4_1} falls under the purview of Proposition \ref{lem:Kar_lem4}, we can infer \eqref{lem2_2shw1}.

\textit{Step 6 :} 
As a consequence of Proposition \ref{lem:Kar_lem4}, there exists $R_1 > 0$ such that $|\gmo{t}| < R_1 / (t+1)^{\delta_0}$ for all $0 \leq \delta_0 < \alpha_1 - \mu_1 $. We choose $\delta_0 \leq \min\{\alpha_1 - \mu_1, \theta_1 - \alpha_1 \}$, which ensures $(1+t)^{-\theta_1} \leq (1+t)^{-(\alpha_1 + \delta_0)}$. Thus for sufficiently large $t$ we have - 
\begin{equation} \label{lem4_2}
    |\gmt{t+1}| \leq \bbc{1 - (1 - 2s) \alpha(t) } |\gmt{t}| + \frac{ \alpha_0 R_1}{(t+1)^{\alpha_1 + \delta_0}} 
\end{equation}  
As \eqref{lem4_2} falls under the purview of Proposition \ref{lem:Kar_lem2}, we have 
\begin{equation*}
    \lim_{\tti} (t+1)^{\delta'_0} \gmt{t} = 0
\end{equation*} 
for all $0 \leq \delta'_0 < \delta_0 $. 
\\ By making $\delta_0$ arbitrarily close to $\alpha_1 - \mu_1$ we get \eqref{lem2_2shw2}.

\end{proof}

Let us define a new matrix $J$ as $J \coloneqq I - \J$. 
\\ The following lemma provides a bound for $\nrm{J - \beta(t)L^\infty}$.  
\begin{lemma} \label{lem:NormBnd}
    Given $c_1 > 0$, $\linf = \bbc{\dout - A} \winf$ where $\winf = \text{diag}\bbc{w^\infty}$, $\beta(t) = \frac{\beta_0}{(1+t)^{\beta_1}}$, $\mu(t) = \frac{\mu_0}{(1+t)^{\mu_1}}$ where $0 < \beta_0 < \psi$, $\mu_0 > 0$ and $0 < \beta_1 < \mu_1 < 1$, there exists  $T > 0$ such that $\nrm{J - \beta(t)\linf} \leq 1 - c_1 \mu(t) < 1$ for all $t \geq T$. 
\end{lemma}
\begin{proof}
Using the property $\one^T \linf = 0 $ we can write
\begin{equation} \label{lem6_1}
\begin{split} 
    &\nrmsq{J - \beta(t)\linf} = \lambda_{\text{max}} \left( J - \beta(t) M_2 + \beta^2(t) M_3 \right) \\ 
     &= \sup_{x \in \R^N, ||x||=1} x^T \left( J - \beta(t) M_2 + \beta^2(t) M_3 \right) x
\end{split}    
\end{equation}
where $M_2 = \bbc{\linf^T  + \linf }$ and $M_3 = \linf^T \linf$.
Now by definition $\linf \one = 0$. So we have   
\begin{equation} \label{lem6_2}
    M_2 \one = 0 \text{  ;  } M_3 \one = 0
\end{equation}
Also, $M_2$ and $M_3$ are the Laplacians of the corresponding graph. As the graphs are strongly connected, we can infer : 
\begin{equation*} 
    \lambda_2 \bbc{M_2} > 0 \text{  ;  } \lambda_2 \bbc{M_3} > 0 
\end{equation*} 
i.e the 2nd lowest eigen-value of each of the Laplacians is strictly positive.  
\\ Let $x \in \text{span}\{\one\} \equiv x = \alpha \one, \alpha \in \R $. Then using \eqref{lem6_2} we have 
\begin{equation} \label{lem6_3}
    x^T \left( I - \J - \beta(t) M_2 + \beta^2(t) M_3 \right) x = 0
\end{equation}
Now suppose $x \in \text{span}\{\one^\perp \} \equiv x^T \one = 0 $. Then we have  
\begin{equation} \label{lem6_4}
\begin{split}
    x^T \left( I - \J -\beta(t) M_2 + \beta^2(t) M_3 \right) x \\ 
     \leq \bbc{1 - \bbc{\beta(t) \lambda_m - \beta^2(t) \lambda_M}} x^T x
\end{split}
\end{equation}
where $\lambda_m = \lambda_2 (M_2)$ and $\lambda_M = \lambda_\text{max} (M_3)$. 
\\ Having $\beta_0 < \psi$ and $\beta_1 < 1$ ensures that $\beta(t) < \lambda_m / \lambda_M$, and in turn $\nrmsq{I - \J - \beta(t) \linf } < 1$ $\forall t \geq 0$.
We choose an $\epsilon$ such that $0 < \epsilon \leq \lambda_m - \beta(t) \lambda_M$. Then for all $t \geq 0$
\begin{equation} \label{lem6_5}
    \beta(t) \lambda_m - \beta^2(t) \lambda_M \geq \beta(t) \epsilon > 0
\end{equation}
Then from \eqref{lem6_1}, \eqref{lem6_3}, \eqref{lem6_4} and \eqref{lem6_5} we have  
\begin{equation} \label{eq:NormBnd_1}
    \nrm{J - \beta(t) \linf} \leq \sqrt{1 - \beta(t) \epsilon} < 1
\end{equation}
Now, as $\mu_1 > \beta_1$, there exists time $T > 0$ such that for all $t > T$ 
\begin{equation*}
\begin{split}
    \frac{1}{(1+t)^{\mu_1 - \beta_1}} \leq \frac{\epsilon \beta_0}{2 c_1 \mu_0}  
    \implies 2 c_1 \mu(t) \leq \epsilon \beta(t) \\ 
    \implies 1 - \epsilon \beta(t) \leq 1 - 2 c_1 \mu(t) + c_1^2 \mu^2(t)
\end{split}
\end{equation*} 
\begin{equation} \label{eq:NormBnd_2}
    \therefore \sqrt{1 - \epsilon \beta(t)} \leq 1 - c_1 \mu(t)
\end{equation} 
Also as $c_1 > 0$, we have $1 - c_1 \mu(t) < 1$. 
Then from \eqref{eq:NormBnd_1} and \eqref{eq:NormBnd_2} we have 

$\nrm{J - \beta(t) \linf} \leq 1 - c_1 \mu(t) < 1 \text{ , } t \geq T$
\end{proof} 

\section{} \label{sec:Apdx_Thm}
\begin{proof} [Proof of Theorem \ref{thm}] Let $\Bar{x}(t) \in \R^{1 \times M}$ denote the average of the states of the agents : $\Bar{x}(t) = (1/N)\one^T x(t)$. 

\noindent We define $p(t)$ as the difference between the state of the agents and their average, and $q(t)$ as the difference between the average value and the unknown parameter $\ts{t}$.
\begin{align} \label{eq:par_pq}
    p(t) &\coloneqq x(t) - \one \Bar{x}(t) , p(t) \in \R^{N \times M} \\ 
    q(t) &\coloneqq \Bar{x}^T(t) - \ts{t} , q(t) \in \R^M.
\end{align}
Now the norm of the difference between the state of each agent and $\ts{t}$ can be upper bounded as 
\begin{equation} \label{2shw_1}
\begin{split}
    \nrm{x_i(t) - \ts{t}} &\leq \nrm{x_i(t) - \Bar{x}^T(t)} + \nrm{\Bar{x}^T(t) - \ts{t}} \\ 
    \implies \nrm{x_i(t) - \ts{t}} &\leq \nrm{p(t)} + \nrm{q(t)} \text{ for all } i \in \node .
\end{split}
\end{equation} 
To show that the states converge to $\ts{t}$, we express $\nrm{x_i(t) - \ts{t}} \leq \gamma(t)$ and show that $\gamma(t)$ has a dynamics that converges to 0. In what follows, we firstly we use the method of induction to show that 
 \begin{equation} \label{ind_fnl}
     \nrm{p(t)} \leq \gamma_1(t) \text{ and } \nrm{q(t)} \leq \gamma_2(t) \text{ for all } t \geq 0
 \end{equation} 
 and in the process also define the dynamics of $\gmo{t}$ and $\gmt{t}$. After that we express these dynamics as a linear time-varying system which is asymptotically stable. From there, using \eqref{2shw_1}, \eqref{eq:Gam} and \eqref{ind_fnl}  we arrive at our desired result.

\textbf{Dynamics of $x(t) - \one \Bar{x}(t)$ : } 

$p(t+1) = x(t+1) - \one \Bar{x}(t+1) = J x(t+1)$
\begin{equation} \label{p_dyn}
    \implies p(t+1) = M_1 + \alpha(t)JK(t)(y(t) - x(t))
\end{equation} 
where $J \coloneqq I - \J$, and $M_1 \coloneqq J(I - \beta(t)L(t))x(t)$. 
Expanding $M_1$ and using $\one^T L(t) = 0$ and $J \one = 0$, followed by applying norm and its properties of triangle-inequality and sub-multiplicativity we get
\begin{align} \label{m1_nrm1} 
    \nrm{M_1} &\leq \nrm{(J - \beta(t)L_\infty)} \nrm{p(t)} + \beta(t) \big( \nrm{L(t) \one} \nrm{q(t)} \nonumber \\
     &+ \nrm{(L(t) - L_\infty)} \nrm{p(t)} + \nrm{L(t) \one} \nrm{\ts{t}} \big)
\end{align}
Now applying Lemmas \ref{lem:Lt_Linf} and \ref{lem:NormBnd}, and Assumption-\ref{asmp:thtaBnd} in \eqref{m1_nrm1} : 
\begin{equation} \label{m1_nrm2}
    \nrm{M_1} \leq  (1 - c_1 \mu(t)) \nrm{p(t)}
     + C \beta(t) \eta^t \big( \nrm{p(t)} + \nrm{q(t)} +1 + \Theta \big).
\end{equation}
Applying norm to \eqref{p_dyn} and using \eqref{eq:par_Kt}, $\nrm{J}=\one$ we get 
\begin{align} \label{p_nrm2}
    \nrm{p(t+1)} &\leq (1 - c_1 \mu(t) + C \beta(t) \eta^t) \nrm{p(t)}   \\ \nonumber
     &+ C \beta(t) \eta^t \big( 1 + \Theta + \nrm{q(t)} \big) + \sqrt{N} \alpha(t) \gamma(t).
\end{align}

\textbf{Dynamics of $ \Bar{x}^T(t) - \ts{t} $ : } 
\begin{equation} \label{eq:Qt_1}
    \begin{split}
        q(t &+ 1) = \Bar{x}^T(t+1) - \ts{t+1} = \Bar{x}^T(t) - \ts{t} \\
        &+ \frac{\alpha(t)}{N} \one^T K(t) ( y(t) - x(t) ) + \underbrace{\ts{t+1} - \ts{t}}_{\Delta \ts{t+1}}
    \end{split}
\end{equation}
We define two diagonal matrices $\kg(t) , \kb(t)$ where $[\kg (t)]_{ii} \coloneqq k_i(t)$ if $i \in \gd$ and $[\kb (t)]_{ii} \coloneqq k_i(t)$ if $i \in \bd$, and rest of the entries are equal to $0$. 
\begin{equation} \label{eq:KKgKb}
        \therefore \kb(t) + \kg(t) = K(t) \text{  for all } t \geq 0
\end{equation}
Using \eqref{eq:par_yi} and \eqref{eq:par_pq} in \eqref{eq:Qt_1}, followed by applying the $l_2$-norm and its properties of triangle-inequality and sub-multiplicativity we get 
\begin{align} \label{eq:Qt_nrm1}
    \nrm{q(t+1)} \leq \nrm{1 &- \frac{\alpha(t)}{N} \sum_{i \in \gd } k_i(t)} \nrm{q(t)} + \nrm{\Delta \ts{t+1}} \nonumber \\
      &+ \frac{\alpha(t)}{N} \nrm{\sum_{i \in \gd } k_i(t) (x_i(t) - \Bar{x}^T(t) )} \nonumber \\ 
      &+ \frac{\alpha(t)}{N} \nrm{\sum_{i \in \bd } k_i(t) ( y_i(t) - x_i(t) )}   
\end{align} 

\textbf{Dynamics of $\gamma_1(t)$ and $\gamma_2(t)$ via method of Induction : } 
\\ By the method of induction we wish to show \eqref{ind_fnl}, and in the process arrive at the dynamics of $\gamma_1(t)$ and $\gamma_2(t)$.
\\ \textit{Step 1} : at $t = 0$, $\nrm{p(0)} = 0$ as $x(0) = \zro$, and $\nrm{q(0)}=\Theta$ as $\nrm{\ts{t}} < \Theta$. 
Choosing $\gamma_1(0) = 0$, $\gamma_2(0) = \Theta$ we have  
\begin{equation} \label{ind_s1}
    \nrm{p(0)} \leq \gamma_1(0) \text{ , } \nrm{q(0)} \leq \gamma_2(0) 
\end{equation} 
\textit{Step 2} : for some $t>0$ we assume that  
\begin{equation} \label{ind_s2}
    \nrm{p(t)} \leq \gamma_1(t) \text{ , } \nrm{q(t)} \leq \gamma_2(t)
\end{equation}
\textit{Step 3} : based on the assumption \eqref{ind_s2} from Step-2, we need to show that  
\begin{equation} \label{ind_s3}
    \nrm{p(t+1)} \leq \gamma_1(t+1) \text{ , } \nrm{q(t+1)} \leq \gamma_2(t+1)
\end{equation}
Applying \eqref{ind_s2} to \eqref{p_nrm2} and using \eqref{eq:Gam} we have  
\begin{align*}
    \nrm{p(t &+ 1)} \leq  (1 - c_1 \mu(t) + C\beta(t)\eta^t + \sqrt{N} \alpha(t)) \gamma_1(t) \\
     &+ (C \beta(t) \eta^t + \sqrt{N} \alpha(t) ) \gamma_2(t) + C (1 + \Theta) \beta(t) \eta^t
\end{align*}
Now as $\eta < 1$ , there exists $c_2 > 0$ and $T > 0$ such that for all $t > T$ 
\begin{equation} \label{approx_2}
    C (1 + \Theta) \beta(t) \eta^t \leq c_2 \eta^t \text{ , and  } 
    C \beta(t) \eta^t \leq \alpha(t)
\end{equation} 
By appropriate choice of $\beta_0 < \frac{\alpha_0}{C}  $, $c_2 > C (1 + \Theta) \beta_0 $ and $\mu_0 < (\lambda_m - \beta_0 \lambda_M) \beta_0 / (2 c_1)$, \eqref{approx_2} and \eqref{eq:NormBnd_2} holds for all $t > 0$
\begin{equation}\label{p_nrm3} 
\begin{split}
    \therefore  \nrm{p(t+1)} \leq (1 - c_1 \mu(t) + (1+\sqrt{N}) \alpha(t)) \gamma_1(t) \\
    + (1 + \sqrt{N}) \alpha(t) \gamma_2(t) + c_2 \eta^t
\end{split}
\end{equation}
We define the dynamics of $\gmo{t}$ as -
\begin{equation} \label{Gam1}
\begin{split}
    \gmo{t+1} \coloneqq (1 - c_1 \mu(t) + (1+\sqrt{N}) \alpha(t)) \gmo{t} \\ 
    + (1 + \sqrt{N}) \alpha(t) \gmt{t} + c_2 \eta^t 
\end{split}
\end{equation} 
Using \eqref{eq:Gam}, \eqref{2shw_1}, and \eqref{ind_s2}  
\begin{equation} \label{eq:Ki_1}
    k_i(t) = 1 \text{ for all } i \in \gd \text{  [}\because y_i(t) = \ts{t} \forall i \in \gd]
\end{equation} 
Now from  \eqref{eq:KKgKb}, \eqref{eq:Qt_nrm1}, \eqref{eq:Ki_1} and further using \eqref{eq:Gam}, \eqref{eq:thtaVarBnd} and $\alpha_0 < 1/(1 - 2s)$, $s < 1/2$  we get
\begin{equation} \label{q_nrm2}
    \nrm{q(t+1)} \leq (1 - \alpha(t) (1-2s) ) \gmt{t}  + \alpha(t) \gmo{t} + 1 / (1+t)^{\theta_1} 
\end{equation}
We define the dynamics of $\gmt{t}$ as 
\begin{equation} \label{Gam2}
    \gmt{t+1} \coloneqq \alpha(t) \gmo{t} + (1 - \alpha(t) (1-2s) ) \gmt{t} + 1 / (1+t)^{\theta_1} 
\end{equation} 
Then from \eqref{p_nrm3}, \eqref{Gam1} and \eqref{q_nrm2}, \eqref{Gam2} we can infer \eqref{ind_s3}. Thus from steps 1,2 and 3 we have \eqref{ind_fnl}.

\textbf{Asymptotic stability of $\gmo{t}$ and $\gmt{t}$ : } 
\\ Using Lemma \ref{lem:convergence} we can say that a linear time-varying system with state-variables $\gmo{t}$ and $\gmt{t}$, and state dynamics given by \eqref{Gam1} and \eqref{Gam2} respectively, is asymptotically stable, i.e \\  $\lim_{\tti} (t+1)^{\delta_0} \gmo{t} = 0 $, $ \lim_{\tti} (t+1)^{\delta_0} \gmt{t} = 0$.
\end{proof}

\section{} \label{sec:Apdx_psi} 

For our REWB algorithm, the value of a constant $\psi$, which is an upper bound to the parameter $\beta_0$, is defined as $\psi := 2/(N \degInMax (\degInMax + \degOutMax))$. Here we provide a detailed proof of how we arrived at this value of $\psi$.

We have $0 < \beta_0 < \lambda_2(\linf^T + \linf)/\lambda_{\text{max}}(\linf^T \linf)$. Now using Gershgorin's Disk Theorem we can write :
\begin{align*}
    \lambda_{\text{max}} &(\linf^T \linf) \leq 2 \max_i [\linf^T \linf]_{ii} \\ 
     &= 2 \max_i \big( [\linf]^2_{ii} + \sum_{j \in \mathcal{N}_i , j \neq i} [\linf]^2_{ij} \big) \\ 
     &= 2 \max_i \big( (d_i^{\text{out}} w_i^\infty)^2 + \sum_{j \in \mathcal{N}_i , j \neq i} (w_j^\infty)^2 \big) 
\end{align*}
Now using $d_i^{\text{out}} w_i^\infty \leq 1$ from \cite{BalWts}, and $\frac{1}{d_i^{\text{out}}} \leq 1$, we have : 
\begin{equation*}
    \lambda_{\text{max}} (\linf^T \linf) \leq 2 \max_i \big( 1 + \sum_{j \in \mathcal{N}_i , j \neq i} (\frac{1}{d_i^{\text{out}}})^2 \big) \leq 2 d_{\text{max}}^{\text{in}} 
\end{equation*}
From the results in \cite{Mohar}, we can say : $\lambda_2(\linf^T + \linf) > \frac{4}{N(\degInMax + \degOutMax)}$. 
\begin{equation*}
    \therefore \frac{2}{N \degInMax (\degInMax + \degOutMax)} < \frac{\lambda_2(\linf^T + \linf)}{\lambda_{\text{max}}(\linf^T \linf)}
\end{equation*} 
So defining $\psi : = 2/(N \degInMax (\degInMax + \degOutMax))$ and choosing any non-zero positive value of $\beta_0 < \psi$ satisfies the requirement of our REWB algorithm.

\section{} \label{sec:Apdx_Laplacians} 
In the proof for Lemma \ref{lem:NormBnd} in Appendix \ref{sec:Apdx_IL}, we use the fact that the second eigenvalues of the matrices $M_2$ and $M_3$ are non-zero, where $M_2 = \bbc{\linf^T  + \linf }$ and $M_3 = \linf^T \linf$. Here we provide a reasoning for the same. 

\begin{itemize}
\item zero column sum : from \eqref{eq:par_LInf} and \eqref{eq:BalWts} we have $\one^T \linf = \zro, \linf \one = \zro$. Using these, we get
\begin{align*}
    \one^T M_2 = \one^T \linf^T + \one^T \linf = \zro &; \one^T M_3 = \one^T \linf^T \linf = \zro
\end{align*}
\item positive diagonal elements : 
    \begin{align*}
        [M_2]_{ii} &= [\linf^T]_{ii} + [\linf]_{ii} = 2[\linf]_{ii} > 0 \\ 
        [M_3]_{ii} &= [\linf^T]_{i:} [\linf]_{:i} = \sum_{j = 1}^N [\linf]_{ji}^2 > 0 
    \end{align*}
    where $[M]_{ii}, [M]_{i:}, [M]_{:i}$ represent the $(i,i)$-th element, $i$-th row, and $i$-th column of matrix $M$ respectively.
\item non-diagonal elements in $M_2$ : 
    $[M_2]_{ij} = [\linf]_{ji} + [\linf]_{ij}$.
\item non-diagonal elements in $M_3$ : 
\begin{equation*}
    [M_3]_{ij} = [\linf^T]_{i:}[\linf]_{:j} = \sum_{k = 1}^N [\linf]_{ki} [\linf]_{kj}
\end{equation*}
\end{itemize}
From the expression of the entries of $M_2$ and $M_3$, one can deduce that these matrices will have a nonzero entry in the $ij$th position if the digraph has an edge between $i$ and $j$. This shows that the connectivity of the graph corresponding to $\linf$ would be preserved in the new graph corresponding to $M_2$ and $M_3$. Now, as $M_2$ and $M_3$ are valid Laplacians, and their corresponding graphs are connected, we can infer that the second eigenvalues of $M_2$ and $M_3$ are non-zero.

\balance

\printbibliography

\end{document}